\documentclass[a4paper,UKenglish,cleveref, autoref, thm-restate, numberwithinsect]{lipics-v2021}
\usepackage{float}
\usepackage{iftex}
\usepackage{xcolor}

\ifpdf
  \usepackage{underscore}         
  \usepackage[T1]{fontenc}        
\else
  \usepackage{breakurl}           
\fi

\DeclareMathAlphabet{\mathbbm}{U}{bbm}{m}{n}          
\newcommand{\IN}{\mathbbm{N}}     
\definecolor{mycolor}{RGB}{200,0,120}
\usepackage{amsthm} 
\usepackage{diagbox}
\usepackage{amssymb}
\usepackage{amsfonts}
\usepackage{graphicx}
\usepackage{multicol}
\usepackage[matrix,arrow,cmtip,curve]{xy}
\usepackage{amsmath}
\usepackage{tikz,tikz-3dplot,tikz-cd}
\usepackage{caption}
\usetikzlibrary{arrows.meta, positioning} 
\usetikzlibrary{fit, arrows, automata, shapes, calc, fadings,shapes.geometric}
\tikzstyle{startstop} = [rectangle, rounded corners, 
minimum width=3cm, 
minimum height=1cm,
text centered, 
draw=black, 
fill=red!30]

\tikzstyle{io} = [trapezium, 
trapezium stretches=true, 
trapezium left angle=70, 
trapezium right angle=110, 
minimum width=3cm, 
minimum height=1cm, text centered, 
draw=black, fill=blue!30]

\tikzstyle{process} = [rectangle, 
minimum width=3cm, 
minimum height=1cm, 
text centered, 
text width=3cm, 
draw=black, 
fill=orange!30]

\tikzstyle{decision} = [diamond, 
minimum width=3cm, 
minimum height=1cm, 
text centered, 
draw=black, 
fill=green!30]
\tikzstyle{arrow} = [thick,->,>=stealth]
\tikzset{->, auto, >=stealth', font=\small}
\tikzset{state/.style={shape=circle, draw, fill=white, initial text=,
    inner sep=.5mm, minimum size=1.5mm}}

\tikzset{accepting/.style=accepting by arrow}
\tikzset{state with output/.style={shape=rectangle split, rectangle
    split parts=2, draw, fill=white,
    initial text=, inner sep=1mm}}
    
\usepackage{tikz} 
\newcommand*{\Path}{\mathrm{Path}}
\newcommand*{\cell}{\mathrm{Cell}}
\newcommand*\exec{%
  \raisebox{1pt}{%
    \begin{tikzpicture}[x=.8ex,y=1ex,-]
      \draw (0,0) -- (1,0) -- (1,1) -- (2,1);
    \end{tikzpicture}}}
\newcommand*\subsu{\sqsubseteq}
\newcommand*\unsubsu{\sqsupseteq}
\newcommand*\lan[1]{\mathcal{L}_{#1}}
\newcommand*\bis[1]{\approx_{#1}}
\renewcommand{\bis}[1]{\mathrel{\,		      
	\raisebox{.3ex}{$\underline{\makebox[.7em]{$\leftrightarrow$}}$}
                  \,_{#1}}}
\newcommand{\ininc}{\hookrightarrow }
\newcommand{\congruence}{similarity} 
\newcommand{\Congruence}{Similarity} 
\newcommand{\congruent}{similar} 
\renewcommand{\simeq}{\leftrightarrows}
\renewcommand{\triangle}{\Lambda}
\newcommand{\plat}[1]{\raisebox{0pt}[0pt][0pt]{#1}} 

\title{Bisimulations and Modal Logics\texorpdfstring{\newline}{} for Higher Dimensional Automata}
\titlerunning{Bisimulations and Modal Logics for Higher Dimensional Automata}

\author{Safa Zouari}{Norwegian University of Science and Technology Gjøvik, Norway \and \url{https://www.safazouari.com}}{safa.zouari@ntnu.no}{https://orcid.org/0009-0005-4882-5524}{Supported by a PhD fellowship at the Department of Information Security and Communication Technology (IIK), Norwegian University of Science and Technology (NTNU).}
\author{Rob {van Glabbeek}}{School of Informatics, University of Edinburgh, Scotland \and School of Computer Science and Engineering, University of New South Wales, Sydney, Australia \and \url{https://theory.stanford.edu/~rvg/}}{rvg@cs.stanford.edu}{https://orcid.org/0000-0003-4712-7423}{Supported by Royal Society Wolfson Fellowship RSWF\textbackslash R1\textbackslash 221008}
\author{Krzysztof Ziemiański}{Institute of Mathematics, University of Warsaw, Poland}{ziemians@mimuw.edu.pl}{https://orcid.org/0000-0001-7695-4028}{}

\authorrunning{S. Zouari, R.J. van Glabbeek \&  K. Ziemiański}
\Copyright{Safa Zouari, Rob van Glabbeek, and Krzysztof Ziemiański}



\nolinenumbers

\EventEditors{Ana Sokolova and Patrick Totzke}
\EventNoEds{2}
\EventLongTitle{37th International Conference on Concurrency Theory (CONCUR 2026)}
\EventShortTitle{CONCUR 2026}
\ArticleNo{49}

\begin{document}

\maketitle

\begin{abstract}
    Higher-Dimensional Automata (HDAs) provide a geometric model of true concurrency. While hereditary history-preserving (hhp) bisimilarity is the finest behavioural equivalence in van Glabbeek’s spectrum, no modal logic has previously characterised it on HDAs. We introduce several new intermediate equivalences that sit strictly between ST- and hhp-bisimilarity.
We show how separating {\congruence} and subsumption of paths leads to a clean formulation of these equivalences, and we present a modal logic that characterises hhp-bisimilarity. Natural fragments characterise ST-bisimilarity and the intermediate notions.
\end{abstract}

\section{Introduction}
Higher-dimensional automata (HDAs) \cite{PrattCG,vG91} provide a geometric model of true concurrency:
$n$-dimensional cells represent the concurrent execution of $n$ independent events,
without reducing concurrency to interleaving.
HDAs occupy a central position in the spectrum of concurrency models:
Petri nets  \cite{Reisig85}, event structures \cite{winskel1986event}, configuration structures \cite{van1995configuration}, and asynchronous systems
\cite{bednarczyk1989categories} embed into HDAs up to hereditary history-preserving (hhp) bisimilarity~\cite{VANGLABBEEK2006265}.
Thus HDAs form a unifying semantic framework for non-interleaving concurrency.

Hereditary history-preserving bisimilarity (hhp-bisimilarity) is one of the finest
behavioural equivalences, respecting both branching time and causality.
For several true-concurrency models, modal logics have been shown to characterise
hhp-bisimilarity (e.g.\ on  event structures~\cite{BC14},
stable configuration structures~\cite{phillips2014event},
prime event structures and history-dependent automata~\cite{baldan2014hereditary}).
For HDAs, however, despite being introduced 35 years ago,
no modal logic has provided an exact characterisation
of hhp-bisimilarity.
Prisacariu's Higher-Dimensional Modal Logic~\cite{P10} did not characterise even ST-bisimilarity \cite{GV87},
reaching only split bisimulation; a follow-up~\cite{P13} still fell short of hp-bisimilarity.
More recently, Zouari, Ziemia\'{n}ski and Fahrenberg~\cite{ZZF25}, inspired by the open maps technique,
succeeded in characterising ST-bisimilarity, but only in the absence of autoconcurrency.
Our work is the first to provide exact modal characterisations of the full spectrum
from ST- to hhp-bisimilarity on HDAs, including in the presence of autoconcurrency.

\medskip
\noindent
\textbf{Decomposing path adjacency.}
A key technical contribution of this paper is a decomposition of the classical
\emph{$\ell$-adjacency} relation on paths~\cite{VANGLABBEEK2006265}
into two conceptually distinct components.
\emph{Similarity} at a position $\ell$ in a path (written $\simeq^\ell$) permutes two adjacent
independent start phases of events, or two adjacent end phases, and is inherently symmetric.
\emph{Subsumption} (written $\sqsubseteq^\ell$) replaces a sequential traversal
around a higher-dimensional cell by a path that passes through it directly,
increasing the degree of observed concurrency; it is directional.
In the original formulation~\cite{VANGLABBEEK2006265} these two operations
are bundled into a single symmetric adjacency step. This forces any  adjacency-based notion of bisimulation to require either both the forward and backward directions of subsumption, or neither of them.
Making the decomposition explicit allows one to require subsumption in only one
direction, yielding a strictly finer control over what behavioural information
is preserved.

\medskip
\noindent
\textbf{A spectrum of new bisimilarities.}
The decomposition of path adjacency into similarity and subsumption
naturally gives rise to a family of bisimulation equivalences that fills the gap between
ST-bisimilarity and hhp-bisimilarity in the spectrum of~\cite{van2001refinement,vG90}.
This family is parameterised by which subset of the operations
subsumption, unsubsumption, similarity and path restriction
is required of the bisimulation relation.
These intermediate notions are of independent interest:
they isolate distinct behavioural properties of concurrent systems,
allowing one to distinguish systems that differ in their capacity
for concurrency refinement from those that differ in the symmetry
of their concurrent structure.

We establish strict inclusions between these notions using explicit separating examples.
Between two of the equivalences we establish inclusion without a separating
example; hereby we pose the open question whether these equivalences coincide.

\medskip
\noindent
\textbf{A general modal characterisation framework.}
To obtain modal characterisations uniformly across the entire spectrum, we implicitly apply
transformations of HDAs into a kind of labelled transition system called \emph{$(I,M)$-system}.
These have state-labels as well as transition labels. For ST-bisimilarity, hhp-bisimilarity, and each of
the new intermediate equivalences on HDAs, we apply a transformation that translates this equivalence on
HDAs to strong bisimilarity on the associated $(I,M)$-system. 
The states in our $(I,M)$-system will be the paths in the HDA; they are labelled by their associated
ST-trace, which describes the observable content of a path through a sequence of action-phases
(starts and ends of actions), with a matching of each end phase with it corresponding start.
The labelled transitions in the $(I,M)$-system will be the relations in $I \subseteq \{$path
extension, subsumption, similarity, unsubsumption, path restriction$\}$.
The upshot of this is that a form of the classical Hennessy-Milner logic on labelled transition
systems can be used to reason about these path-based semantic equivalences on HDAs.
Each bisimilarity in our spectrum is characterised by an appropriate
choice of the index set $I$, yielding the corresponding modal logic as a
natural fragment of a common language.
This framework also explains, in a uniform way, why the logics characterising
coarser equivalences are syntactic fragments of those for finer ones.

As a consequence of the above method, our logics on HDA are \emph{path-based}, whereas previous
proposals \cite{P10,P13} could be classified as \emph{cell-based}.
\vspace{-1ex}

\paragraph*{Contributions}
\vspace{-1.5ex}
\begin{itemize}
    \item We decompose $\ell$-adjacency into $\ell$-similarity and $\ell$-subsumption,
    identifying these as independent behavioural conditions of conceptually
    different character.

    \item Using this decomposition, we introduce a family of intermediate
    bisimulation equivalences between ST- and hhp-bisimilarity,
    including in particular \emph{semi-history-preserving} (shp)
    and \emph{quasi-history-preserving} (qhp) bisimilarity,
    and establish strict inclusions between them via explicit separating examples.

    \item We introduce $(I,M)$-systems and prove a general Hennessy-Milner
    theorem: bisimilarity and logical equivalence coincide for any choice
    of index set $I$ and label set $M$.

    \item Instantiating this framework to HDAs, we obtain exact modal
    characterisations of each bisimilarity in our spectrum as a natural
    fragment of a single modal language, providing in particular
    the first modal characterisation of hhp-bisimilarity for HDAs.
\end{itemize}

\medskip
\noindent
\textbf{Why a finer spectrum matters.}
Bisimulation-based techniques have a strong track record in the verification of security protocols (e.g.\ for information-flow properties~\cite{FG01}), and can even uncover real-world privacy vulnerabilities~\cite{BHMY23,HMY23,HM21}.
More directly relevant to our setting, Aubert, Horne and Johansen~\cite{AHJ22} show that moving to finer non-interleaving bisimilarity notions expands the class of detectable attacks on privacy protocols, with the underlying transition systems coinciding exactly with HDAs; in~\cite{AHJM26} this line of work explicitly takes advantage of history-preserving bisimilarity.
Our intermediate equivalences, each accompanied by a characterising modal logic, provide a principled way to calibrate the attacker model between ST- and hhp-bisimilarity, potentially uncovering attacks beyond ST without requiring the full strength of hhp.
Developing these connections into concrete verification tools is left as future work.

A \emph{testing scenario} is a method to explain the non-equivalence of two systems
  in terms of an observation that could be made for one but not the other. One of the few testing
  scenarios that distinguishes ST-bisimilar system is ``Teams can see pomsets'' of Plotkin \& Pratt
  \cite{PP96}. It might be that a suitable branching time version of this testing scenario
  corresponds with one of the semantic equivalences explored in our paper.
\vspace{2ex}

\noindent
\textbf{Organisation of the paper.}
Section~\ref{sec:hda} 
recalls the ordered precubical presentation of HDAs. Section~\ref{sec: path and their labels} introduces paths and decomposes $\ell$-adjacency
into $\ell$-similarity and $\ell$-subsumption. Section~\ref{sec: ST-trace} recalls the ST-trace of a path.
Section~\ref{sec: Bis for HDA} introduces the bisimulation equivalences arising
from the decomposition and establishes the
inclusion hierarchy. Section~\ref{sec: bis logic} develops the $(I,M)$-system framework and proves
the general Hennessy-Milner theorem. Section~\ref{sec: Modal char} instantiates that framework to HDAs and derives modal characterisation for each bisimilarity in our spectrum.

\section{Higher-dimensional automata}\label{sec:hda}
Higher-dimensional automata (HDAs) model concurrent computation
geometrically: an $n$-dimensional cell represents the simultaneous
execution of $n$ independent events.
Formally, HDAs are defined as presheaves over a precubical base
category encoding how events start and terminate. In this section we recall the classical ordered precubical
presentation of HDAs, following~\cite{VANGLABBEEK2006265, LanguageofHDA}.
Fix a set of actions $\Sigma$.

\begin{definition}[The ordered precube category $\square$]\rm
Objects are \emph{concurrency lists} (\emph{conclists}) $((n),\lambda)$,
where $(n)=\{1<\dots<n\}$ is a finite totally ordered set
and $\lambda\colon(n)\to\Sigma$ is a labelling. We write
$(n,\lambda)$ instead of $((n),\lambda)$ for simplicity.

A morphism $(f,\varepsilon)\colon(m,\lambda)\to(n,\mu)$ consists of:
\begin{itemize}
    \item an injective order- and label-preserving map
          $f\colon(m)\hookrightarrow(n)$;
    \item a status map $\varepsilon\colon(n)\to\{0,\exec,1\}$
          satisfying $\varepsilon^{-1}(\exec)=f((m))$.
\end{itemize}
\end{definition}
Intuitively, an object $U=(n,\lambda)$ represents a configuration in which $n$ concurrent events are being performed; they are linearly ordered solely in order to distinguish them when relating events from different configurations.
Given $(f,\varepsilon)\colon U\to V$, the status map $\varepsilon$ classifies each
event of $V$ as \emph{not started} ($0$), \emph{executing} ($\exec$),
or \emph{terminated} ($1$), while $f$ embeds the active events of $U$
into those of $V$. Since $f$ is injective and order-preserving,
given $U=(m,\lambda)$ and $V=(n,\mu)$,
it is uniquely determined by the active set
$\varepsilon^{-1}(\exec)\subseteq (n)$.
Additionally, the event labelling $\lambda$ is completely determined by $\mu$ and $f$.

Composition of $(f,\varepsilon)\colon S\to T$
and $(g,\zeta)\colon T\to U$
is defined by
$$
(g,\zeta)\circ(f,\varepsilon)
=(g\circ f,\eta),
$$
where
$$
\eta(u)=
\begin{cases}
\varepsilon(t) & \text{if } u=g(t)\text{ for some } t\in T,\\
\zeta(u) & \text{otherwise}.
\end{cases}
$$

\medskip

For $U=(n\!-\!1,\lambda)$, $V=(n,\lambda')$
and $k\in\{0,1\}$,
the \emph{coface map} $d^k_{i,n}\colon U\to V$
$(1\le i\le n)$ is the unique morphism
$(\iota_i,\varepsilon)$
where $\iota_i:(n{-}1)\rightarrow(n)$ omits position $i$ from $(n)$,
$\varepsilon(i)=k$,
and $\varepsilon(j)=\exec$ for all $j\neq i$.
When $n$ is clear we write $d^k_i$.
Intuitively, when $V$ is seen as an $n$-dimensional cube, representing the concurrent execution of $n$ events, $U$ represents its $n{-}1$-dimensional face in which the $i^{\rm th}$ event of $V$ has not yet started (if $k{=}0$) or is already finished (if $k{=}1$).
The coface maps satisfy the cubical identities
$$
d_{j,n}^{\ell} \circ d_{i,n-1}^{k} = d_{i,n}^{k} \circ d_{j-1,n-1}^{\ell}
\qquad (i<j),\quad k,\ell \in \{0,1\},
$$
and generate the category $\square$, in the sense that any morphism can be obtained by composing coface maps.
A \emph{precubical set} is a \emph{presheaf over $\square$}, that is, a functor
$$
X\colon \square^{\it op}\to \mathbf{Set}.
$$
If \( X \) is a presheaf over  \(\square\) and \( U=(n,\lambda) \) is an object, $X[U]$ represents the set of those configurations in $X$ in which exactly $n$ events are active, with labels $(\lambda(i))_{i=1}^n$.
Elements of \( X[U] \), for all objects $U$ of $\square$, form the set of cells \(\cell_X\) of \( X \). Specifically, we have:
$$
\cell_X = \bigsqcup_{U \in \text{obj}(\square)} X[U].
$$
Let $\square\left(U, V \right)$ denote the set of morphisms from $U$ to $V$. If the dimension of $U$ is one lower than that of $V$, these are coface maps.
For a coface map $d_i^k$ we write $\delta_i^k := X[d_i^k]$; it is a function that for each $n$-dimensional cell yields its starting or end face (depending on $k$) in dimension $i$.

Let $0$ denote the unique object $(n,\lambda)$ in $\square$ with $n=0$.

\begin{definition}[Higher-dimensional automaton]\label{def:hda}\rm
An \emph{HDA} is a pair $\mathcal X=(X,i_{\mathcal X})$
where
$X$ is a precubical set
and $i_{\mathcal X}\in X(0)$
is a distinguished initial $0$-cell.
\end{definition}

\section{Paths in HDAs}\label{sec: path and their labels}
We recall the notion of paths in HDAs, fixing notation for later use.
\begin{definition}\label{def:path}\rm
 A \emph{path} of length $n$ in a precubical set $X$ is a sequence\\ $
\alpha=\left(x_{0}, \varphi_{1}, x_{1}, \varphi_{2}, \ldots, \varphi_{n}, x_{n}\right),$
where $x_{j} \in X\left[U_{j}\right]$ are cells, and for all $j\in\{1,\dots,n\}$, either
\begin{itemize}
  \item $\varphi_{j}=d_{i_{j}}^{0} \in \square\left(U_{j-1}, U_{j}\right)$, a coface map, and $x_{j-1}=\delta_{i_{j}}^{0}\left(x_{j}\right)$ (up-step), or
  \item $\varphi_{j}=d_{i_{j}}^{1} \in \square\left(U_{j}, U_{j-1}\right), \delta_{i_{j}}^{1}\left(x_{j-1}\right)=x_{j}$ (down-step).\vspace{1ex}
\end{itemize}
A \emph{path in an HDA} $\mathcal{X}=(X,i_{\mathcal X})$ is a path in the
underlying precubical set $X$ whose first cell is the initial cell $i_{\mathcal X}$.
We write $\Path_X$ for the sets of all paths in $X$, and $\Path_{\mathcal{X}}\subseteq\Path_X$ for the sets of paths starting at the initial cell.
\end{definition}
\noindent
If $\alpha=(x_{0},\varphi_{1},\ldots,\varphi_{n},x_{n})$ and $\beta=(y_{0},\psi_{1},\ldots,\psi_{m},y_{m})$ are paths in $X$ with $x_n=y_0$, their
\emph{concatenation} is the path
$$
\alpha*\beta=
\bigl(x_{0},\varphi_{1},\ldots,\varphi_{n},x_{n},
\psi_{1},y_{1},\ldots,\psi_{m},y_{m}\bigr).
$$
In this situation, we say that $\alpha * \beta$ is an
\emph{extension} of $\alpha$ and that $\alpha$ is a \emph{restriction} or
\emph{prefix} of $\alpha * \beta$; we write $\alpha \ininc \alpha*\beta$.

\begin{definition}[{\Congruence} of paths]\label{def:path-congruence}\rm
We say that paths $\alpha$ and $\beta$ are $\ell$-\emph{\congruent}, written $\alpha\simeq^{\ell}\beta$, if one is obtained from the other by replacing
\[
    \text{(1)}\quad (d_i^0,\;x_\ell,\;d_j^0)\quad \text{by} \quad(d^0_{j-1},\;x'_\ell,\;d_i^0),
    \qquad \text{or} \qquad
    \text{(2)}\quad (d_j^1,\;x_\ell,\;d_i^1)\quad  \text{by} \quad(d_i^1,\;x'_\ell,\;d_{j-1}^1)
\]
for indices $i<j$.
\emph{\Congruence} $\simeq$ is the equivalence relation generated by $\simeq^\ell$ for all $\ell\geq 1$.
\end{definition}
\begin{definition}[Subsumption of paths]\rm
\label{def:ell-subsumption}
We say that path \emph{$\beta$ subsumes $\alpha$ at position $\ell$}, written $\alpha \subsu^{\ell} \beta$, if $\alpha$ is obtained from $\beta$ by replacing
$$
\begin{aligned}
\text{(3)}\quad
&(d_i^0,\,x_\ell,\,d_j^1)\;~\mbox{by}~\;(d^1_{j-1},\,x'_\ell,\,d_i^0), \quad \text{ or } \quad 
\text{(4)}\quad(d_j^0,\,x_\ell,\,d_i^1)\;~\mbox{by}~\;(d_i^1,\,x'_\ell,\,d_{j-1}^0)
\end{aligned}
$$ 
for indices $i<j$.
The \emph{global subsumption} relation $\subsu$ is the reflexive and transitive
closure of all $\ell$-subsumptions:
$$
\alpha\;\subsu\;\beta
\quad\Longleftrightarrow\quad
\exists\,\ell_1,\dots,\ell_k\text{ such that }
\alpha=\alpha^{(0)}\subsu^{\ell_1}\alpha^{(1)}\subsu^{\ell_2}\cdots
\subsu^{\ell_k}\alpha^{(k)}=\beta.
$$
\end{definition}
Geometrically, if $\alpha \subsu^\ell \beta$, then $\beta$ passes
through the $n$-dimensional cell $x_\ell$,
while $\alpha$ bypasses it by following the corresponding $(n-2)$-dimensional face $x'_{\ell}$ instead.
In particular, $\beta$ explores a region of the automaton with more concurrency; see Fig.~\ref{fig:subsumption-geometry} for the $n=2$ case.

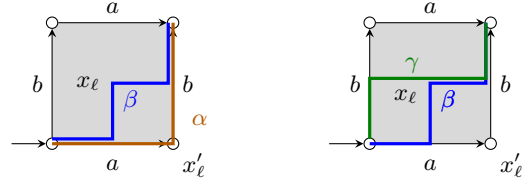
\begin{figure}[t]
  \centering
  \begin{tikzpicture}[x=1.6cm, y=1.6cm, >=stealth]

    \begin{scope}[name prefix=L-, xshift=0cm]
      \path[fill=black!15] (0,0) -- (1,0) -- (1,1) -- (0,1) -- cycle;

      \node[state, initial] (00) at (0,0) {};
      \node[state] (10) at (1,0) {};
      \node[state] (01) at (0,1) {};
      \node[state] (11) at (1,1) {};

      \path (00) edge node[below] {$\vphantom{d}a$} (10);
      \path (01) edge node[above] {$a$} (11);
      \path (00) edge node[left] {$b$} (01);
      \path (10) edge node[right] {$b$} (11);

      \node at (0.3,0.5) {$x_\ell$};
      \node[below right] at (1,0) {$x'_\ell$};

      \draw[-, very thick, blue] (0,0.04) -- (0.5,0.04) -- (0.5,0.5) -- (0.96,0.5) -- (0.96,1);
      \node[blue] at (0.65,0.35) {$\beta$};

      \draw[-, very thick, orange!70!black] (0,0) -- (1,0) -- (1,1);
      \node[orange!70!black] at (1.22,0.2) {$\alpha$};
    \end{scope}

    \begin{scope}[name prefix=R-, xshift=4.2cm]
      \path[fill=black!15] (0,0) -- (1,0) -- (1,1) -- (0,1) -- cycle;

      \node[state, initial] (00) at (0,0) {};
      \node[state] (10) at (1,0) {};
      \node[state] (01) at (0,1) {};
      \node[state] (11) at (1,1) {};

      \path (00) edge node[below] {$\vphantom{d}a$} (10);
      \path (01) edge node[above] {$a$} (11);
      \path (00) edge node[left] {$b$} (01);
      \path (10) edge node[right] {$b$} (11);

      \node at (0.3,0.39) {$x_\ell$};
      \node[below right] at (1,0) {$x'_\ell$};

    
      \node[blue] at (0.65,0.35) {$\beta$};
\node[green!50!black] at (0.5-.15,0.65) {$\gamma$};
 \draw[-, very thick, blue] (0,0) -- (0.5,0) -- (0.5,0.5) -- (0.96,0.5) -- (0.96,1);
      \node[blue] at (0.65,0.35) {$\beta$};
\draw[-, very thick, green!50!black] (0,0.04) -- (0,0.54) -- (0.46,0.54) -- (0.96,0.54) -- (0.96,1);
    \end{scope}

  \end{tikzpicture}
\vspace{-4pt}
\caption{Subsumption step as a bypass of a higher-dimensional cell (illustrated for $n=2$).
  Here \textcolor{blue}{$\beta$} could be a path
  \textcolor{blue}{$(x_0,d^0_{1,1},x_1,d^0_{2,2},x_2,d^1_{1,2},x_3,d^1_{1,1})$} in which $x_0,x_1,x_2,x_3,x_4$ are
  cells of dimensions $0,1,2,1,0$, respectively, $d^0_{i,n}$ indicates the start and $d^1_{i,n}$ the
  end of an action, numbered $i$ out of $n$. Here $\ell=2$. Now \textcolor{orange!70!black}{$\alpha$}
  becomes \textcolor{orange!70!black}{$(x_0,d^0_{1,1},x_1,d^1_{1,1},x'_2,d^0_{1,1},x_3,d^1_{1,1})$}, where
  $x'_2$ is a cell of dimension 0. In this case $\beta \unsubsu^2 \alpha$, so $\beta$ subsumes $\alpha$;
  in $\beta$ the two actions overlap, whereas in $\alpha$ they don't. Moreover, \textcolor{green!50!black}{$\gamma$} could be a path
  \textcolor{green!50!black}{$(x_0,d^0_{1,1},x'_1,d^0_{1,2},x_2,d^1_{1,2},x_3,d^1_{1,1})$}, where
  $x'_1$ is a cell of dimension 1. Here $\gamma \simeq^{r} \beta$ with $r=1$, so that $\gamma$ and $\beta$ are similar paths.}
  \label{fig:subsumption-geometry}
\vspace{-4pt}
\end{figure}

\begin{definition}[Adjacency of paths \cite{VANGLABBEEK2006265}]\rm
\label{def:adjacency} 
Paths $\alpha$ and $\beta$ are $\ell$-adjacent, written $\alpha\overset{\ell}{\leftrightsquigarrow} \beta$, if they are either $\ell$-similar or $\ell$-subsumed. 
\emph{Homotopy} is the equivalence relation generated by similarity and subsumption. This is, two paths $\alpha$ and $\beta$ are \emph{homotopic} if one can be obtained from the other by a sequence of $\ell$-similarities and $\ell$-subsumptions. 
\end{definition}

Let $\alpha=(x_0,\varphi_1,x_1,\varphi_2,\dotsc,\varphi_m,x_m)$ be a path of length $m$ and $1\leq \ell<m$. If $\varphi_{\ell}$ and $\varphi_{\ell+1}$ are both up-steps or both down-steps, then there exists a unique path $\alpha^{(\ell)}$ that is $\ell$-similar to $\alpha$. If $\varphi_{\ell}=d^0_i$ is an up-step and $\varphi_{\ell+1}=d^1_j$ is a down-step, then there exists a unique $\alpha^{(\ell)}$ such that $\alpha^{(\ell)}\subsu^\ell\alpha$ if $i\neq j$, and no $\ell$-adjacent paths otherwise.
In all these cases we have \plat{$\alpha \overset{\ell}{\leftrightsquigarrow} \alpha^{(\ell)}$}.
Finally, if $\varphi_\ell$ is a down-step and $\varphi_{\ell+1}$ is an up-step, there may be any number (including none)  of paths $\beta$ such that $\alpha\subsu^\ell\beta$.

In the original work \cite{VANGLABBEEK2006265}, $\ell$-{\congruence} and $\ell$-subsumption
are bundled into the single symmetric $\ell$-adjacency relation.
In this paper, we decompose $\ell$-adjacency into two conceptually distinct components:
  \emph{$\ell$-\congruence} and $\ell$-\emph{subsumption}.
This separation is not merely cosmetic.
While {\congruence} is inherently symmetric, subsumption is directional and cannot be treated symmetrically without loss of structure.
Making this distinction explicit is a key technical step of the present work and directly enables our novel formulation of bisimilarity.
In particular, it allows us to define several variants of hereditary history-preserving bisimilarity (in Section \ref{sec: Bis for HDA}) in which subsumption is required to hold only in a single direction.

\section{ST-trace}\label{sec: ST-trace}
Fix a down-step $(x_\ell,d^1_{i_{\ell+1}},x_{\ell+1})$ occurring in a path 
$$\alpha=(x_0,\varphi_1,x_1,\ldots,\varphi_\ell,x_\ell,d^1_{i_{\ell+1}},x_{\ell+1},
\varphi_{\ell+2},\ldots,x_m)$$
in an HDA\@.
By the text above Definition 21 in \cite{VANGLABBEEK2006265}, there exists a unique index
$k\le \ell$ such that $$\alpha \overset{\ell}{\leftrightsquigarrow} \alpha^{(\ell)} 
\overset{\ell-1}{\leftrightsquigarrow} \alpha^{(\ell-1)} 
\overset{\ell-2}{\leftrightsquigarrow} \cdots 
\overset{k+1}{\leftrightsquigarrow} \alpha^{(k+1)} 
\overset{k}{\not\leftrightsquigarrow} \alpha^{(k)},$$
where each step is obtained by applying the unique adjacency replacement
at the indicated position.

This means that the down-step can be moved successively towards the
beginning of the path by adjacency replacements in the sense of
Definition~\ref{def:adjacency}, until it sits immediately after
its matching up-step and cannot be moved further.

More explicitly, the down-step $d^1_{i_{\ell+1}}$ can be commuted one
position to the left precisely when the adjacent segment
$(d^k_{i_\ell},x_\ell,d^1_{i_{\ell+1}})$ satisfies $k=1$ or
matches one of the subsumption
patterns of Definition~\ref{def:ell-subsumption}, that is, when $i_\ell \neq i_{\ell+1}$.

The process stops exactly at position $k$ because this condition fails:
in $\alpha^{(k+1)}$ the adjacent segment has the form
$(d^0_{i_k},x_{k+1},d^1_{i_k})$, where the indices coincide and no
adjacency replacement is applicable.
We then write $\mathsf{start}(\ell{+}1)=k$.
It says that the down-step $d^1_{i_{\ell+1}}$ at position $\ell{+}1$ marks the end of an event that started with the up-step $d^0_{i_k}$ at position $k$.

\def\ST{\mathsf{ST}}
\def\iP{\mathsf{iPom}}
\def\split{\mathsf{split}}

\begin{definition}[ST-trace \cite{VANGLABBEEK2006265}]\label{def:ST-trace}\rm
Let 
$
\alpha=(x_{0},
d^{\,k_{1}}_{i_{1}},x_{1},
\dots,$ $
d^{\,k_{n}}_{i_{n}},x_{n})\in \Path_{\mathcal{X}}$, and let $\lambda(i_j)$ denote the label of the
event whose start or termination is represented by the face map
$d^{k_j}_{i_j}$.
For each step define
$$\sigma_j^{ST} :=
\begin{cases}
  + & \text{if } k_j = 0,\\
  start(j) & \text{if } k_j = 1.
\end{cases}$$
Then $\ST(\alpha)
  = (\lambda(i_1)^{\sigma^{ST}_1},\ldots,\lambda(i_n)^{\sigma^{ST}_n}).$
\end{definition}
The \emph{ST-trace} $\ST(\alpha)$ records not only when each action starts ($+$) and
terminates, but also links each termination to the position of its
corresponding start via the index $start(j)$.  In this way, the
ST-trace encodes the causal pairing between the beginning and ending of
individual actions, capturing their overlap and nesting
(see~\cite{VANGLABBEEK2006265} for the original formulation).

\begin{figure}
    \centering
    \begin{tikzpicture}[scale=2]
    \def\eps{0.04}
        \path[fill=black!15] (0,0) -- (3,0) -- (3,1) -- (2,1) -- (2,2) -- (1,2) -- (1,1) -- (0,1);
        \node[state] (00) at (0,0) {};
        \node[state] (10) at (1,0) {};
        \node[state] (20) at (2,0) {};
        \node[state] (30) at (3,0) {};
        \node[state] (01) at (0,1) {};
        \node[state] (11) at (1,1) {};
        \node[state] (21) at (2,1) {};
        \node[state] (31) at (3,1) {};
        \node[state] (12) at (1,2) {};
        \node[state] (22) at (2,2) {};
        \draw[->] (00) edge node[below] {$b$} (10);
        \draw[->] (10) edge node[below] {$a$} (20);
        \draw[->] (20) edge node[below] {$b$} (30);
        \draw[->] (00) edge node[left] {$a$} (01);
        \draw[->] (10) -- (11);
        \draw[->] (20) -- (21);
        \draw[->] (30) -- (31);
        \draw[->] (01) -- (11);
        \draw[->] (11) -- (21);
        \draw[->] (21) -- (31);
        \draw[->] (11) edge node[left] {$b$} (12);
        \draw[->] (12) -- (22);
        \draw[->] (21) -- (22);
        \draw[very thick, orange!70!black] (0.03, \eps) -- (0.5,\eps) -- (0.5,0.5) -- (1.5,0.5) -- (1.5,1.5) -- (2,1.5);
        \draw[very thick, blue] (00) -- (0,0.5+\eps) -- (1.5-\eps,0.5+\eps) -- (1.5-\eps,1.5+\eps) -- (2,1.5+\eps);
        \draw[very thick, green!50!black] (0.03, 0) -- (0.5+\eps, 0) -- (0.5+\eps, 0.5-\eps) -- (2.5, 0.5-\eps) -- (2.5,1);
        \node[green!50!black] at (2.2, 0.35) {$\alpha$};
        \node[blue] at (0.3, 0.65) {$\gamma$};
        \node[orange!70!black] at (1.6, 0.8) {$\beta$};
    \end{tikzpicture}
    \caption{Examples of paths in an HDA.}
    \label{fig:path-example}
\end{figure}

\begin{example}
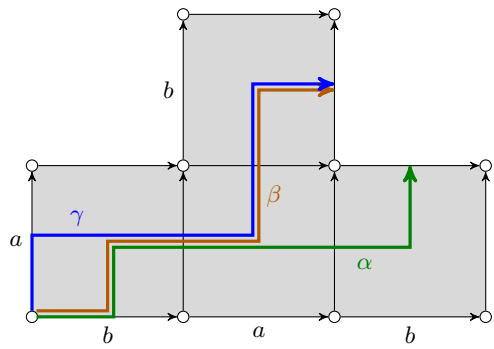
The paths depicted in Fig.\@ \ref{fig:path-example} have all different ST-traces:
\[
\ST(\alpha)=b^+a^+b^1a^+a^4b^+a^2,
\quad
\ST(\beta)=b^+a^+b^1a^+a^2b^+a^4,
\quad
\ST(\gamma)=a^+b^+b^2a^+a^1b^+a^4.
\]
\end{example}

\section{Bisimulations for Higher Dimensional Automata}\label{sec: Bis for HDA}
In this section, we review the main notions of bisimilarity defined for HDAs. Additionally, we introduce several new bisimilarities that lie between history-preserving bisimilarity and ST-bisimilarity in the spectrum established in \cite{van2001refinement}.

Hereditary history preserving bisimilarity (hhp-bisimilarity) and history preserving bisimilarity (hp-bisimilarity) are equivalence notions that were originally introduced for other concurrency models \cite{van2001refinement,rabinovich1988behavior,bednarczyk1991hereditary}, and were later adapted to higher-dimensional automata (HDA) \cite{VANGLABBEEK2006265}. 
Within the settings of this work, (h)hp-bisimilarity is defined as follows. 
\begin{definition}\rm \label{def: hp bisimilarity}\rm
 A \emph{history-preserving bisimulation} (\emph{hp-bisimulation})
between HDAs $\mathcal Y$ and $\mathcal Z$
is a symmetric relation $\mathfrak R$
between paths of $\mathcal Y$ and $\mathcal Z$ such that:
\begin{enumerate}
\item \label{en: h bis initial path} \emph{Initial condition.} $(i_{\mathcal Y},i_{\mathcal Z})\in\mathfrak R$. Here $i_{\mathcal Y}$ is the unique path in $\mathcal Y$ of length $0$.
\item \label{en: h bis pomset iso} \emph{Label equality}. If $(\alpha,\beta)\in\mathfrak R$, then
$\ST(\alpha) =  \ST(\beta)$.
\item \label{en: h bis concatenation} \emph{Path extension.}
If $(\alpha,\beta)\in\mathfrak R$ and $\alpha\ininc \alpha'$ (i.e.\ $\alpha$ is an initial subpath of $\alpha'$),
then there exists $\beta'$ such that $\beta\ininc \beta'$ and $(\alpha',\beta')\in\mathfrak R.$
\item \label{en: h bis subsu} \emph{Subsumption}.\!
If $(\alpha,\beta)\mathbin\in\mathfrak R$
and $\alpha \mathbin{\subsu^{\ell}} \alpha'\!$,
then there exists $\beta'$ such that $\beta \mathbin{\subsu^{\ell}} \beta'$
and $(\alpha'\!,\beta')\mathbin\in\mathfrak R$.

\item \label{en: h bis unsubsu} \emph{Unsubsumption.}
If $(\alpha,\beta)\mathbin\in\mathfrak R$
and $\alpha' \mathbin{\subsu^{\ell}} \alpha$,
then there exists $\beta'$ with $\beta' \mathbin{\subsu^{\ell}} \beta$ and $(\alpha',\beta')\mathbin\in\mathfrak R$.

\item \label{en: h bis cong} \emph{\Congruence.}
If $(\alpha,\beta)\mathbin\in\mathfrak R$
and $\alpha' \mathbin{\simeq^{\ell}} \alpha$,
then there exists $\beta'$ such that $\beta' \mathbin{\simeq^{\ell}} \beta$ and $(\alpha',\beta')\mathbin\in\mathfrak R$.
\end{enumerate}
The relation $\mathfrak R$ is called
a \emph{hereditary history-preserving bisimulation}
(\emph{hhp-bisimulation})
if, in addition, it satisfies:
\begin{enumerate}
\setcounter{enumi}{6}
\item \label{en: h bis hereditary} \emph{Path restriction.}
If $(\alpha,\beta)\mathbin\in\mathfrak R$
and $\alpha' \mathbin{\ininc} \alpha$,
then there exists $\beta' \mathbin{\ininc} \beta$
such that $(\alpha'\!,\beta')\mathbin\in\mathfrak R$.\vspace{1ex}
\end{enumerate}
We say that $\mathcal X $ and $\mathcal Y$ are \emph{(hereditary) history preserving bisimilar} (\emph{(h)hp-bisimilar}), denoted $\mathcal X \bis{(h)hp} \mathcal Y$, if there exists a (hereditary) history-preserving bisimulation ((h)hp-bisimulation) between them. This clearly is an equivalence relation.
\end{definition}
We will address the clauses 3--7 of Definition~\ref{def: hp bisimilarity} by the symbols
  $\hookrightarrow$, $\subsu$, $\unsubsu$, $\simeq$ and $\hookleftarrow$, respectively.
An \emph{ST-bisimulation} between HDAs $\mathcal Y$ and $\mathcal Z$ is a symmetric relation $\mathfrak R$ between paths of $\mathcal Y$ and $\mathcal Z$ satisfying Clauses~\ref{en: h bis initial path}--\ref{en: h bis concatenation} of Definition~\ref{def: hp bisimilarity}. 
We say that $\mathcal X $ and $\mathcal Y$ are \emph{ST-bisimilar}, denoted $\mathcal X \bis{ST} \mathcal Y$, if an ST-bisimulation between them exists. So by definition ${\bis{hhp}} \subseteq {\bis{hp}} \subseteq {\bis{ST}}$.

From Definition~\ref{def: hp bisimilarity} we derive $2^4$
bisimulation equivalences that are included between $\bis{ST}$ and $\bis{hhp}$.
Namely for any subset $I \subseteq \{\subsu, \unsubsu, \simeq, \hookleftarrow\}$ we define an $I$-bisimulation as a relation between the paths of two HDAs that satisfies Clauses 1, 2 and 3 of Def.~\ref{def: hp bisimilarity}, as well as the clauses from $I$;
two HDAs are $I$-bisimilar iff an $I$-bisimulation between them exists.
It is trivial to check that all these notions are in fact equivalence relations.

Thus $\emptyset$-bisimilarity is ST-bisimilarity, $\{\subsu, \unsubsu, \simeq,
\hookleftarrow\}$-bisimilarity is hhp-bisimilarity and $\{\subsu, \unsubsu, \simeq\}$ is
hp-bisimilarity. We give names to six more of these equivalence notions that may yield new and different equivalences on HDAs: \emph{semi-\!-history-preserving (shp)
bisimilarity} ($\bis{shp}$) has $I=\{\subsu\}$ and \emph{quasi-\!-history-preserving (qhp)
bisimilarity} ($\bis{qhp}$) has $I=\{\subsu, \simeq\}$; in both cases we get \emph{hereditary}
versions ($\bis{hshp}$ and $\bis{hqhp}$) by adding clause $\hookleftarrow$.
We denote $I=\{\simeq, \hookleftarrow\}$-bisimilarity as $\bis{reST}$ and
$I=\{\simeq,\hookleftarrow\}$-bisimilarity as $\bis{ureST}$; here the letters $r$, $e$ and $u$ stand
for ``path \emph restriction'', ``similarity \emph equivalence'' and ``\emph unsubsumption''.

The following table summarises which conditions are satisfied by each of these bisimilarities.
\begin{center}
\renewcommand{\arraystretch}{1.2}
\begin{tabular}{|>{\raggedleft\arraybackslash}p{3.28cm}|c|@{~}c@{~}|@{~}c@{~}|c|@{~}c@{~}|c|@{~}c@{~}|c|c|}

  \hline
  \textbf{Condition} 
  & $\bis{ST}$ 
  & $\bis{\it reST}$
  & $\bis{\it ureST}$
  & $\bis{\it shp}$ 
  & $\bis{\it hshp}$ 
  & $\bis{\it qhp}$ 
  & $\bis{\it hqhp}$ 
  & $\bis{\it hp}$ 
  & $\bis{\it hhp}$\hspace{-.4pt} \\ 
  \hline
  Initial paths are related 
  & $\checkmark$ & $\checkmark$ & $\checkmark$& $\checkmark$ & $\checkmark$ & $\checkmark$ & $\checkmark$ & $\checkmark$ & $\checkmark$ \\ 
  \hline
  Labels of related paths
  & $\ST $ & $\ST $ & $\ST $ & $\ST $ & $\ST $ & $\ST $ & $\ST $ & $\ST $ & $\ST $ \\ 
  \hline
  Extension 
  & $\checkmark$ & $\checkmark$ & $\checkmark$ & $\checkmark$ & $\checkmark$ & $\checkmark$ & $\checkmark$ & $\checkmark$ & $\checkmark$ \\ 
  \hline
  Subsumption  
  &  &  & & $\checkmark$ & $\checkmark$ & $\checkmark$ & $\checkmark$ & $\checkmark$ & $\checkmark$ \\ 
  \hline
  {\Congruence}
  &  & $\checkmark$ & $\checkmark$ & & & $\checkmark$ & $\checkmark$ & $\checkmark$ & $\checkmark$  \\ 
  \hline
  Unsubsumption  
  &  & & $\checkmark$ &  &  & & & $\checkmark$ &  $\checkmark$ \\ 
  \hline
  Path restriction
  &  & $\checkmark$ & $\checkmark$ &  & $\checkmark$ &  & $\checkmark$ &  & $\checkmark$ \\ 
  \hline

\end{tabular}
  \label{tab:conditions}
\end{center}

\noindent
 Now we establish the relationships between the various notions of equivalence that were defined above.
Our findings are summarised in Fig.\ref{fig:lattice}.

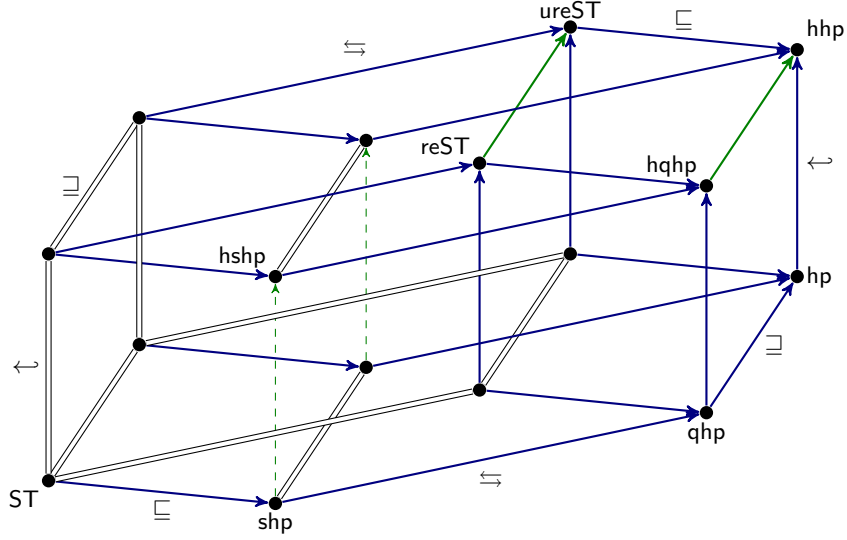
\begin{figure}
    \centering
    \begin{tikzpicture}[scale=3]
    \def\rad{5pt}
    \def\sx{1}
    \def\sy{-0.1}
    \def\ux{0.4}
    \def\uy{0.6}
    \def\rx{0}
    \def\ry{1}
    \def\ex{1.9}
    \def\ey{0.4}
    \def\ndst{\node[fill,circle,minimum size=\rad,inner sep=0pt]}
    \def\drdk{\draw[dashed,green!50!black,->]}
    \def\drin{\draw[thick,blue!50!black,->]}
    \def\drnin{\draw[thick,green!50!black,->]}
    \def\dreq{\draw[double,double distance=1.5pt,-]}
    \def\drxx{\draw[thick,red,->]}
    \ndst (ST) at (0,0) {};
    \ndst (s) at (\sx,\sy) {};
    \ndst (u) at (\ux,\uy) {};
    \ndst (su) at (\sx+\ux,\sy+\uy) {};
    \ndst (r) at (\rx,\ry) {};
    \ndst (sr) at (\sx+\rx,\sy+\ry) {};
    \ndst (ur) at (\ux+\rx,\uy+\ry) {};
    \ndst (sur) at (\sx+\ux+\rx,\sy+\uy+\ry) {};
    \ndst (e) at (\ex,\ey) {};
    \ndst (se) at (\sx+\ex,\sy+\ey) {};
    \ndst (ue) at (\ux+\ex,\uy+\ey) {};
    \ndst (sue) at (\sx+\ux+\ex,\sy+\uy+\ey) {};
    \ndst (re) at (\rx+\ex,\ry+\ey) {};
    \ndst (sre) at (\sx+\rx+\ex,\sy+\ry+\ey) {};
    \ndst (ure) at (\ux+\rx+\ex,\uy+\ry+\ey) {};
    \ndst (sure) at (\sx+\ux+\rx+\ex,\sy+\uy+\ry+\ey) {};
    \node at (\sx*0.5, \sy*0.5-0.08) {$\subsu$};
    \node at (\sx*0.5+\ux+\rx+\ex, \sy*0.5+\uy+\ry+\ey+0.08) {$\subsu$};
    \node at (\ux*0.5+\rx-0.1, \uy*0.5+\ry) {$\sqsupseteq$};
    \node at (\ux*0.5+\sx+\ex+0.1, \uy*0.5+\sy+\ey) {$\sqsupseteq$};
    \node at (\rx*0.5-0.1, \ry*0.5) {$\hookleftarrow$};
    \node at (\sx+\ux+\ex+\rx*0.5+0.1, \sy+\uy+\ey+\ry*0.5) {$\hookleftarrow$};
    \node at (\sx+\ex*0.5, \sy+\ey*0.5-0.1) {$\simeq$}; 
    \node at (\ux+\rx+\ex*0.5, \uy+\ry+\ey*0.5+0.1) {$\simeq$}; 
    \drin (ST) -- (s);
    \drin (u)--(su);
    \drin (r)--(sr);
    \drin (ur)--(sur);
    \drin (e)--(se);
    \drin (ue)--(sue);
    \drin (re)--(sre);
    \drin (ure) -- (sure);
    \dreq (ST)--(u);
    \dreq (s)--(su);
    \dreq (r)--(ur);
    \dreq (sr)--(sur);
    \dreq (e)--(ue);
    \drin (se)--(sue);
    \drnin (re)--(ure);
    \drnin (sre)--(sure);
    \dreq (ST)--(r);
    \drdk (s)--(sr);
    \dreq (u)--(ur);
    \drdk (su)--(sur);
    \drin (e)--(re); 
    \drin (se)--(sre);
    \drin (ue)--(ure);
    \drin (sue)--(sure);
    \dreq (ST)--(e);
    \drin (s)--(se);
    \dreq (u)--(ue);
    \drin (su)--(sue); 
    \drin (r)--(re); 
    \drin (sr)--(sre); 
    \drin (ur)--(ure);
    \drin (sur)--(sure);
    \node[below left] at (ST) {\textsf{ST}};
    \node[above left] at (re) {\textsf{reST}};
    \node[above] at (ure) {\textsf{ureST}};
    \node[below] at (s) {\textsf{shp}};
    \node[below] at (se) {\textsf{qhp}};
    \node[right] at (sue) {\textsf{hp}};
    \node[above right] at (sure) {\textsf{hhp}};
    \node[above left] at (sr) {\textsf{hshp}};
    \node[above left] at (sre) {\textsf{hqhp}};
    \end{tikzpicture}
    \caption{Hierarchy of bisimulations on HDAs. Blue arrows point to finer equivalences, double lines denote coinciding equivalences, dashed green lines are conjectured coincidences. Green arrows point to finer equivalences that are not proven in this paper.}
    \label{fig:lattice}
\end{figure}
 
\begin{theorem}\rm If $\mathcal X \bis{ST} \mathcal{Y}$ then $\mathcal X$ and $\mathcal Y$ are related by a $\{\hookleftarrow\}$-bisimulation, that is, an ST-bisimulation that also satisfies path restriction.
\end{theorem}
\begin{proof}
Assume that there exists an ST-bisimulation $\mathfrak K$ between $\mathcal{X}$ and
$\mathcal{Y}$. Let $\mathfrak{R}$ be the relation defined by $$(\alpha,\beta)\in \mathfrak{R} \Longleftrightarrow (\alpha',\beta')\in \mathfrak K \text{ for all $ \alpha ',\beta'$ such that } \alpha ' \hookrightarrow \alpha,~ \beta ' \hookrightarrow \beta \text{ and } |\alpha'|=|\beta'|.$$
Note that in Clause 3 (path extension) of Def.~\ref{def: hp bisimilarity} one can without limitation of generality restrict to the case where $|\alpha'| = |\alpha|+1$, that is, $\alpha'$ extends $\alpha$ by exactly one pair $(\varphi,x)$ of a coface map and a cell. Using this, all Clauses 1--3 of Def.~\ref{def: hp bisimilarity} hold for $\mathfrak R$ because they hold for $\mathfrak{K}$. Moreover, Clause 7 (path restriction) holds by construction.
\end{proof}

\begin{theorem}\rm\label{STun}
If $\mathcal X \bis{ST} \mathcal{Y}$ then $\mathcal X$ and $\mathcal Y$ are related by a $\{\unsubsu,\simeq\}$-bisimulation, that is, an ST-bisimulation that also satisfies unsubsumption and similarity.
\end{theorem}
\begin{proof}
Let $\mathfrak K\subseteq \Path_{\mathcal X}\times\Path_{\mathcal Y}$ be an ST-bisimulation between $\mathcal X$ and $\mathcal Y$. Define a relation $\mathfrak R\subseteq \Path_{\mathcal X}\times \Path_{\mathcal Y}$ as the least relation with
\(
  {\mathfrak K}\subseteq \mathfrak R
\)
that satisfies the following closure conditions:
\begin{align*}
(\text{Similarity})\quad & (\alpha,\beta)\in {\mathfrak R},\ \alpha\simeq^{\ell} \alpha',\ \beta\simeq^\ell \beta'\ \Longrightarrow\ (\alpha',\beta')\in \mathfrak R,\\
  (\text{Unsubsumption})\quad & (\alpha,\beta)\in {\mathfrak R},\ \alpha'\subsu^{\ell}\alpha,\ \beta'\subsu^{\ell}\beta\ \Longrightarrow\ (\alpha',\beta')\in \mathfrak R.
\end{align*}
It suffices to show that $\mathfrak R$ is an ST-bisimulation. Clause~1.\ (the initial paths are related) is valid by construction since
${\mathfrak K}\subseteq \mathfrak R$. The remaining conditions need to be shown for each pair $(\alpha,\beta)\in\mathfrak R$; we will do so by induction on the derivation of
$(\alpha,\beta)\in\mathfrak R$ from the closure conditions above. The induction base is that 
$(\alpha,\beta)\in\mathfrak K$, and in this case the conditions hold because they hold for
$\mathfrak K$. For the induction step, we may assume that the required conditions have been shown already to hold for the pair $(\alpha,\beta)\in\mathfrak R$, and now consider a pair 
$(\alpha',\beta')\in\mathfrak R$ obtained by one of the above closure conditions.

For Clause~2.\ (preservation of ST-traces), note that if 
$\alpha\simeq^{\ell} \alpha'$, then $\ST(\alpha')$ is completely determined by
$\ell$ and $\ST(\alpha)$, and the same applies when $\alpha'\subsu^{\ell}\alpha$.
Hence from the induction hypothesis that $\ST(\alpha)=\ST(\beta)$ one immediately derives that $\ST(\alpha')=\ST(\beta')$.

For Clause 3 (path extension), let $\alpha' \hookrightarrow \alpha' * \gamma$. Then 
$\alpha \hookrightarrow \alpha * \gamma$, so by Clause 3 for the pair $(\alpha,\beta)\in \mathfrak R$, there exists a path $\delta$ such that $\beta \hookrightarrow \beta * \delta$ and  $(\alpha * \gamma, \beta*\delta)\in \mathfrak R$. Now $\beta' * \delta$ is an extension of $\beta'$ and $(\alpha '* \gamma, \beta' *\delta)\in \mathfrak R$ by the closure properties of $\mathfrak R$.

It remains to show that $\mathfrak R$ satisfies Clauses 5 and 6 (unsubsumption and similarity). So suppose that $(\alpha,\beta)\in\mathfrak R$ and $\alpha' \unsubsu^\ell \alpha$. This means that right before position $\ell$ in the path $\alpha$ one event started and right after $\ell$ another event ended, while in $\alpha'$ these two action phases are reversed. Now since $\ST(\beta)=\ST(\alpha)$, in $\beta$ the situation is similar, so there must be a $\beta'$ such that $\beta' \unsubsu^\ell \beta$. The closure property of $\mathfrak R$ for unsubsumption yields that $(\alpha',\beta')\in\mathfrak R$.
The case for similarity goes likewise.
\end{proof}
The above two theorems imply that $\{\unsubsu\}$-, $\{\simeq\}$-, $\{\unsubsu,\simeq\}$- and $\{\hookleftarrow\}$-bisimilarity coincide with ST-bisimilarity.
Interestingly, although adding to ST-bisimilarity one of the clauses $\simeq$ or 
$\hookleftarrow$ makes no difference, adding both of them yields a strictly finer equivalence ($\bis{reST}$). This will be shown by Thm.~\ref{xST}.

\begin{theorem}\rm\label{shp-uns} If $\mathcal X \bis{shp} \mathcal{Y}$ then $\mathcal X$ and $\mathcal Y$ are related by a $\{\subsu,\unsubsu\}$-bisimulation,
  that is, an shp-bisimulation that also satisfies unsubsumption.
\end{theorem}
\begin{proof}
Let $\mathfrak K\subseteq \Path_{\mathcal X}\times\Path_{\mathcal Y}$ be an shp-bisimulation between $\mathcal X$ and $\mathcal Y$. Define $\mathfrak R$ by closing $\mathfrak K$ under unsubsumption, just as in the proof of Thm.~\ref{STun}.
That $\mathfrak R$ satisfies Clauses 1--3 and 5 (unsubsumption) of Def.~\ref{def: hp bisimilarity} follows exactly as in the proof of Thm.~\ref{STun}.

For Clause 4 (subsumption), let $\alpha' \subsu^k \alpha''$. Recall that by induction the pair $(\alpha,\beta)\in\mathfrak R$ is supposed to satisfy subsumption, and we now want to show this for the pair $(\alpha',\beta')\in\mathfrak R$ obtained using that
$\alpha' \subsu^\ell \alpha$ and $\beta' \subsu^\ell \beta$. The special case that $k=\ell$ yields $\alpha''=\alpha$, so picking $\beta'':=\beta$ finishes the argument.
Now consider the case $k\neq \ell$.
Given that $\unsubsu$ moves a start-phase of an event past the termination of another event, while $\subsu$ does the opposite, the modifications $\alpha \sqsupseteq^\ell \alpha' \subsu^k \alpha''$ do not interfere with each other, i.e., $|k-\ell|\geq 2$, so there must be a path $\alpha'''$ such that $\alpha \subsu^k \alpha''' \sqsupseteq^\ell \alpha''$.
Since $(\alpha,\beta)\in\mathfrak R$ satisfies subsumption, there exists a path $\beta'''$ such that $\beta \subsu^k \beta'''$ and $(\alpha''',\beta''')\in\mathfrak R$.
Since $\mathfrak R$ is closed under unsubsumption, there is a $\beta''$ such that
$\beta''' \unsubsu^\ell \beta''$ and  $(\alpha'',\beta'')\in\mathfrak R$.
Using that $|k-\ell|\geq 2$, there must be a $\beta^\dagger$ such that
$\beta\unsubsu^\ell \beta^\dagger \subsu^k \beta''$. When $\beta\unsubsu^\ell \beta^\dagger$,
there is in fact no choice for $\beta^\dagger$; it is completely determined by $\beta$ and $\ell$. Hence $\beta^\dagger=\beta$, which finishes the proof.
\end{proof}

\begin{theorem}\rm If $\mathcal X \bis{\{\subsu,\hookleftarrow\}} \mathcal{Y}$ then $\mathcal X \bis{\{\subsu,\unsubsu,\hookleftarrow\}} \mathcal{Y}$.
\end{theorem}
\begin{proof}
Let $\mathfrak K\subseteq \Path_{\mathcal X}\times\Path_{\mathcal Y}$ be an  $\{\subsu,\hookleftarrow\}$-bisimulation between $\mathcal X$ and $\mathcal Y$. Define $\mathfrak R$ by closing $\mathfrak K$ under unsubsumption, just as in the proof of Thm.~\ref{shp-uns}.
That $\mathfrak R$ satisfies Clauses 1--5 of Def.~\ref{def: hp bisimilarity} follows exactly as in the proof of Thm.~\ref{shp-uns}.

For Clause 7 (path restriction), let $\alpha''\hookrightarrow \alpha'$. Here $\alpha'\subsu^{\ell}\alpha$ and $\beta'\subsu^{\ell}\beta$, and by induction the pair $(\alpha,\beta)\in\mathfrak R$ is supposed to satisfy path restriction.
Let $\beta''$ be the unique path such that $|\alpha''|=|\beta''|$ and $\beta''\hookrightarrow\beta$. We need to show that $(\alpha'',\beta'')\in\mathfrak R$.

In case $|\alpha''| <  \ell$, meaning that we restrict $\alpha'$ to below the place where it differs from $\alpha$, we have $\alpha''\hookrightarrow \alpha$. Hence there exists a 
$\beta^\dagger\hookrightarrow \beta$ such that $(\alpha'',\beta^\dagger)\in \mathfrak R$.
Since $|\beta''| = |\beta^\dagger| <  \ell$ we have $\beta\dagger = \beta''$, so we are done.

In case $|\alpha''| >  \ell$, there exists an $\alpha^\dagger\hookrightarrow \alpha$ with $\alpha''\subsu^{\ell}\alpha^\dagger$. By induction there is a $\beta^\dagger\hookrightarrow \beta$ with $(\alpha^\dagger,\beta^\dagger)\in\mathfrak R$, and since $\mathfrak R$ is closed under unsubsumption there is a $\beta^\ddagger$ such that $\beta^\ddagger\subsu^{\ell}\beta^\dagger$ and $(\alpha'',\beta^\ddagger)\in\mathfrak R$. It must be that $\beta^\dagger = \beta''$, so again we are done.

The remaining case that $|\alpha''| = \ell$ is the most interesting.
Let $\alpha'''\hookrightarrow \alpha'$ and $\beta'''\hookrightarrow \beta'$ be such that $|\alpha'''| = |\beta'''| = \ell-1$. By the above, $(\alpha''',\beta''')\in \mathfrak R$.
Since $\alpha''' \hookrightarrow \alpha''$ there exists a $\beta''' \hookrightarrow \beta^\dagger$ such that $(\alpha'',\beta^\ddagger)\in\mathfrak R$, using Clause 4. We also have $\beta''' \hookrightarrow \beta''$ and $|\beta^\ddagger| = |\beta''|$. Since the extension of $\beta^\ddagger$ w.r.t.\ $\beta'''$ is the end-phase of a well-defined event occurring in $\beta'''$, it is uniquely determined by $\beta'''$. Hence $\beta^\dagger = \beta''$.
\end{proof}

\begin{theorem}\rm If $\mathcal X \bis{\{\hookleftarrow\}} \mathcal{Y}$ then
   $\mathcal X \bis{\{\unsubsu,\hookleftarrow\}} \mathcal{Y}$.
\end{theorem}
\begin{proof}
This follows by a straightforward simplification of the proof above.
\end{proof}

\begin{theorem}\label{th: shp st inclusion}\rm
  ${\bis{shp}} \subsetneq {\bis{ST}}$. In fact, even
  ${\bis{shp}}\not\supseteq {\bis{\it ureST}}$.
\end{theorem}
\begin{proof}
The inclusion follows directly from the definitions:
every shp-bisimulation satisfies all clauses required for ST-bisimilarity,
since ST-bisimilarity can be obtained from shp-bisimilarity by dropping
Clause~\ref{en: h bis subsu} (subsumption).

Its strictness is established by the example of the HDAs $\mathcal{Y}_1$ and $\mathcal{Y}_2$ illustrated in Figure~\ref{fi: ST bisimilar HDA}. For ST-bisimilarity,
the relation $\mathfrak R$ given by
$$
(\alpha,\beta)\in\mathfrak R\quad\Longleftrightarrow\quad \ST(\alpha)=\ST(\beta).
  $$
is easily seen to be an ST-bisimulation, in that it satisfies the first three conditions of Def.~\ref{def: hp bisimilarity}.
In fact, it is even a $\{\unsubsu,\simeq,\hookleftarrow\}$-bisimulation.
Note that it relates the highlighted paths $\alpha_1$ and $\alpha_2$.
Thus $\mathcal{Y}_1$ and $\mathcal{Y}_2$ are ST-bisimilar.

By definition, if an shp-bisimulation $\mathfrak R$ exists between $\mathcal{Y}_1$ and $\mathcal{Y}_2$, then $(\alpha_1, \alpha_2) \in \mathfrak{R}$. This is because the initial paths of $\mathcal{Y}_1$ and $\mathcal{Y}_2$ must be related, and if we extend the initial path of $\mathcal{Y}_2$ to $\alpha_2$, the only path in $\mathcal{Y}_1$ with the same ST-trace is $\alpha_1$.
However, there exists a path $\beta \in \Path_{Y_1}$, indicated in blue, such that $\alpha_1 \subsu \beta$, while no path $\beta' \in \Path_{Y_2}$ satisfies $\alpha_2 \subsu \beta'$, which violates condition~\ref{en: h bis subsu} of an shp-bisimulation.
\end{proof}
   \begin{figure}[t]
    \centering
    \begin{tikzpicture}[x=1.5cm, y=1.5cm]
     \begin{scope}
       \path[fill=black!15] (0,0) to (1,0) to (1,1) to (0,1);
      \node[state, initial] (00) at (0,0) {};
      \node[state] (10) at (1,0) {};
      \node[state] (01) at (0,1) {};
      \node[state] (11) at (1,1) {};
      \path (00) edge node[below] {$\vphantom{d}a$} (10);
      \path (01) edge node[above] {$a$} (11);
      \path (00) edge node[left] {$b$} (01);
      \path (10) edge node[right] {$b$} (11);
        \draw[-, very thick, blue] (0,0.04) -- (0.5,0.04) -- (0.5,0.5) --  (0.96,0.5) -- (0.96,1);
      \node[blue] at (0.65,0.65) {$\beta$};
         \draw[-, very thick, orange] (0,0) -- (1,0) -- (1,1) ;
           \node[] at (1.2,0.24) {\textcolor{orange}{$\alpha_1$}};
    \end{scope}
      \begin{scope}[shift={(3.5,0)}]
       \path[fill=black!15] (0,0) to (1,0) to (1,1) to (0,1);
      \node[state, initial] (00) at (0,0) {};
      \node[state] (10) at (1,0) {};
      \node[state] (01) at (0,1) {};
      \node[state] (11) at (1,1) {};
      \node[state] (20) at (2,0) {};
      \path (00) edge node[below] {$\vphantom{d}a$} (10);
      \path (10) edge node[below] {$\vphantom{d}b$} (20);
      \path (01) edge node[above] {$a$} (11);
      \path (00) edge node[left] {$b$} (01);
      \path (10) edge node[right] {$b$} (11);
       \draw[-,  very thick, orange] (0,0) --  (2,0) ;
       \node[] at (1.7,0.14) {\textcolor{orange}{$\alpha_2$}};
      \end{scope}
    \end{tikzpicture}
     \caption{Two HDA $\mathcal{Y}_1$ on the left and $\mathcal{Y}_2$ on the right that are ST-bisimilar.}
  \label{fi: ST bisimilar HDA}
  \end{figure}
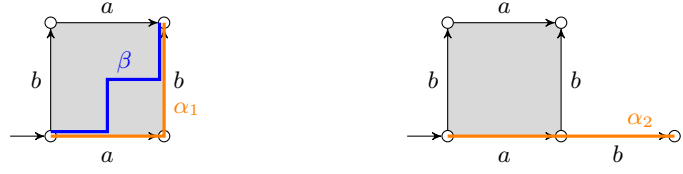

\begin{theorem}\rm
${\bis{qhp}} \subsetneq {\bis{shp}}$. In fact, even ${\bis{qhp}} \not\supseteq {\bis{\{\subsu,\unsubsu,\hookleftarrow\}}}$.
\end{theorem}
\begin{proof}
The inclusions follow directly from the semantics of the relations in question.
\pagebreak[3]

To prove strictness, consider the HDAs $\mathcal{X}$ and $\mathcal{Y}$ in
Fig.~\ref{fig:projected HDA}.
Suppose, towards a contradiction, that there exists an qhp-bisimulation $\mathfrak R$
between $\mathcal{X}$ and $\mathcal{Y}$.
By Clause~\ref{en: h bis initial path}, $(i_{\mathcal X},i_{\mathcal Y})\in\mathfrak R$.
  Now $i_{\mathcal Y}$ can be extended to the path $\beta$, indicated in orange in Fig.~\ref{fig:projected HDA}. So by Clauses~\ref{en: h bis concatenation} and~\ref{en: h bis pomset iso},
there must be a path $\alpha$ in $\mathcal{X}$ such that $(\alpha,\beta)\in\mathfrak R$
and $\ST(\alpha)=a^+c^+a^1c^2b^+d^+b^5d^6$. The only two candidates, $\alpha_1$ and $\alpha_2$, are indicated in orange and purple respectively.

In $\mathcal Y$ we have a subsumption $\beta\subsu^4\beta'$, indicated in blue.
Clause~\ref{en: h bis subsu} therefore forces the existence of a path $\alpha'$ in $\mathcal X$ such that $\alpha\subsu^4\alpha'$ and $(\alpha',\beta')\in\mathfrak R$.
However, in $\mathcal X$ the path $\alpha_2$ has no subsumption successor of this form (again by inspection of Fig.~\ref{fig:projected HDA}), and hence is disqualified as candidate for $\alpha$.

In $\mathcal Y$ we have a similarity $\beta\simeq^3\beta''\simeq^5\beta'''$, obtained by swapping
$a^-$ and $c^-$, as well as $b^+$ and $d^+$, and a subsumption $\beta'''\subsu^4\beta''''$, where
$\beta''''$ is indicated in green in Fig.~\ref{fig:projected HDA}. Clauses~\ref{en: h bis cong}
and~\ref{en: h bis subsu} therefore force the existence of a paths $\alpha'',\alpha''',\alpha''''$ in $\mathcal X$ such that $\alpha_1\simeq^3\alpha''\simeq^5\alpha'''\subsu^4\alpha''''$ and $(\alpha'''',\beta'''')\in\mathfrak R$.
However, whereas $\alpha''$ and $\alpha'''$ can be found, there is no such $\alpha''''$.  This contradiction shows
that $\mathcal X \not\bis{qhp} \mathcal Y$.

 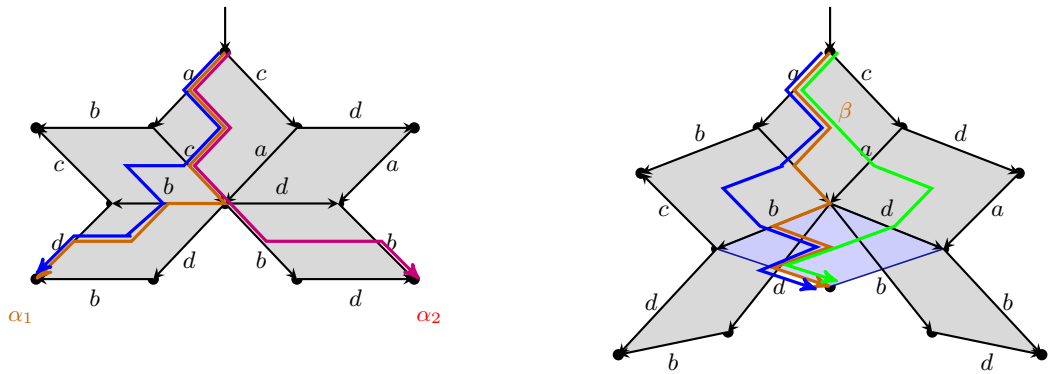
\begin{figure}[ht]
    \centering
    \begin{tikzpicture}
        \begin{scope}[shift={(-4,0)}] 
         \tikzset{
        point/.style={circle, fill, inner sep=1.5pt},
        dot/.style={circle, fill, inner sep=0.8pt},
        arrow/.style={->, >=stealth, line width=0.8pt},
        label/.style={font=\small},
    }
\coordinate (M3) at (0, 2);
    \coordinate (I) at (0, 0);
    \coordinate (E) at (-1.5, 0);
    \coordinate (F) at (1.5, 0); 

    \coordinate (A1) at (-2.5, 1);
    \coordinate (A2) at (-0.5-.45, 1);
    \coordinate (B1) at (-2.5, -1);
    \coordinate (B2) at (-0.5-.45, -1);

    \coordinate (C1) at (2.5, 1);
    \coordinate (C2) at (0.5+.45, 1);
    \coordinate (D1) at (2.5, -1);
    \coordinate (D2) at (0.5+.45, -1);

    \foreach \p in {(I), (E), (F), (A1), (A2), (B1), (B2), (C1), (C2), (D1), (D2),(M3)}{
        \node[point] at \p {};
    }
    \draw[arrow] (0,2.6) -- (M3);
  \path[fill=black!15]  (E) to (A1) to (A2) to (M3) to (C2) to (C1) to (F) to (D1) to (D2) to (I) to  (B2) to (B1) ;

    \draw[arrow] (I) -- (E) node[midway, above] {$b$};
    \draw[arrow] (A2) -- (I) node[midway, above] {$c$};
    \draw[arrow] (I) -- (B2) node[midway, below] {$d$};
\draw[arrow] (M3) -- (C2) node[midway, above] {$c$};
\draw[arrow] (M3) -- (A2) node[midway, above] {$a$};
    \draw[arrow] (I) -- (F) node[midway, above] {$d$};
    \draw[arrow] (C2) -- (I) node[midway, above] {$a$};
    \draw[arrow] (I) -- (D2) node[midway, below] {$b$};

    \draw[arrow] (E) -- (A1) node[midway,  left] {$c$};
    \draw[arrow] (A2) -- (A1) node[midway, above] {$b$};
    \draw[arrow] (E) -- (B1) node[midway, left] {$d$};
    \draw[arrow] (B2) -- (B1) node[midway, below] {$b$};

    \draw[arrow] (C1) -- (F) node[midway,  right] {$a$};
    \draw[arrow] (C2) -- (C1) node[midway, above] {$d$};
    \draw[arrow] (F) -- (D1) node[midway, right] {$b$};
    \draw[arrow] (D2) -- (D1) node[midway, below] {$d$};
\path let \p1=(F), \p2=(C1), \p3=(C2), \p4=(I) in
  coordinate (urTrapCenter) at ({(\x1+\x2+\x3+\x4)/4},{(\y1+\y2+\y3+\y4)/4});

 \coordinate (NEW) at (0, -1.1);
    \coordinate (midMAtwo) at ($(M3)!0.5!(A2)$); 
    \coordinate (upperCenter) at ($(M3)!0.5!(I)$); 
    \coordinate (midA2I) at ($(A2)!0.5!(I)$); 
    \coordinate (midIE) at ($(I)!0.5!(E)$); 
    \coordinate (trapCenter) at ($0.25*(NEW)+0.25*(F)+0.25*(I)+0.25*(E)$); 
    \coordinate (midNE) at ($(NEW)!0.5!(E)$); 
\path let \p1=(I), \p2=(F), \p3=(D1), \p4=(D2) in
  coordinate (drTrapCenter) at ({(\x1+\x2+\x3+\x4)/4},{(\y1+\y2+\y3+\y4)/4});

    \coordinate (midC) at ($(M3)!0.5!(C2)$);  

    \path let \p1=(E), \p2=(A1), \p3=(A2), \p4=(I) in
      coordinate (ulTrapCenter) at ({(\x1+\x2+\x3+\x4)/4},{(\y1+\y2+\y3+\y4)/4});

    \path let \p1=(F), \p2=(C1), \p3=(C2), \p4=(I) in
      coordinate (urTrapCenter) at ({(\x1+\x2+\x3+\x4)/4},{(\y1+\y2+\y3+\y4)/4});

    \path let \p1=(I), \p2=(E), \p3=(B1), \p4=(B2) in
      coordinate (dlTrapCenter) at ({(\x1+\x2+\x3+\x4)/4},{(\y1+\y2+\y3+\y4)/4});

    \path let \p1=(I), \p2=(F), \p3=(D1), \p4=(D2) in
      coordinate (drTrapCenter) at ({(\x1+\x2+\x3+\x4)/4},{(\y1+\y2+\y3+\y4)/4});

\coordinate (midEBone) at ($(E)!0.5!(B1)$);

  \draw[- >, very thick, orange!80!black]  (M3) -- (midMAtwo) -- (upperCenter) -- (midA2I)  -- (I) -- (midIE) -- (dlTrapCenter) -- (midEBone) -- (B1);

  \begin{scope}[transform canvas={xshift=-2pt}]
    \draw[-, very thick, blue] (M3) -- (midMAtwo) -- (upperCenter) -- (midA2I) -- (ulTrapCenter) -- (midIE) -- ([yshift=2pt]dlTrapCenter);
  \end{scope}
   \begin{scope}[transform canvas={yshift=2pt}]
    \draw[->, very thick, blue](dlTrapCenter) -- (midEBone) -- (B1);
  \end{scope}
  
   \begin{scope}[transform canvas={xshift=2pt}]
    \draw[->, very thick, mycolor] (M3) -- (midMAtwo) -- (upperCenter) -- (midA2I) -- (I) -- ($(I)!0.5!(D2)$) -- (drTrapCenter) -- ($(F)!0.5!(D1)$) -- (D1);
  \end{scope}

     \node[font=\footnotesize] at (2.7,-1.5) {\textcolor{red}{$\alpha_2$}};
  \node[font=\footnotesize] at (-2.7,-1.5) {\textcolor{orange!80!black}{$\alpha_1$}};
    
      \end{scope}

        \begin{scope}[shift={(4,0)}]
       \tikzset{
        point/.style={circle, fill, inner sep=1.5pt},
        dot/.style={circle, fill, inner sep=0.8pt},
        arrow/.style={->, >=stealth, line width=0.8pt},
        label/.style={font=\small},
    }
\coordinate (M3) at (0, 2);
    \coordinate (I)   at (0, 0);
    \coordinate (NEW) at (0, -1.1);
    \coordinate (E)   at (-1.5, 0-.6);
    \coordinate (F)   at (1.5, 0-.6); 

    \coordinate (A1) at (-2.5, 1-.6);
    \coordinate (A2) at (-0.5-.45, 1);
    \coordinate (B1) at (-2.8, -1-1);
    \coordinate (B2) at (-0.5-.45-.4, -1-.7);

    \coordinate (C1) at (2.5, 1-.6);
    \coordinate (C2) at (0.5+.45, 1);
    \coordinate (D1) at (2.8, -1-1);
    \coordinate (D2) at (0.5+.45+.4, -1-.7);

    \coordinate (midMAtwo) at ($(M3)!0.5!(A2)$); 
    \coordinate (upperCenter) at ($(M3)!0.5!(I)$); 
    \coordinate (midA2I) at ($(A2)!0.5!(I)$); 
    \coordinate (midIE) at ($(I)!0.5!(E)$); 
    \coordinate (trapCenter) at ($0.25*(NEW)+0.25*(F)+0.25*(I)+0.25*(E)$); 
    \coordinate (midNE) at ($(NEW)!0.5!(E)$); 
    \coordinate (midC) at ($(M3)!0.5!(C2)$);  
   \path let \p1=(I), \p2=(E), \p3=(B1), \p4=(B2) in
      coordinate (dlTrapCenter) at ({(\x1+\x2+\x3+\x4)/4},{(\y1+\y2+\y3+\y4)/4});

\coordinate (midEBone) at ($(E)!0.5!(B1)$);
 \path let \p1=(E), \p2=(A1), \p3=(A2), \p4=(I) in
      coordinate (ulTrapCenter) at ({(\x1+\x2+\x3+\x4)/4},{(\y1+\y2+\y3+\y4)/4});
    \foreach \p in {(I), (E), (F), (NEW), (A1), (A2), (B1), (B2), (C1), (C2), (D1), (D2), (M3)}{
        \node[point] at \p {};
    }
    \draw[arrow] (0,2.6) -- (M3);

    \path[fill=black!15] (E) to (A1) to (A2) to (M3) to (C2) to (C1) to (F) to (D1) to (D2) to (I) to (B2) to (B1);

    \path[fill=blue!20, draw=blue!60!black, line width=0.6pt, opacity=0.9]
      (NEW) -- (F) -- (I) -- (E) -- cycle;

    
    \draw[arrow] (I) -- (E) node[midway, above] {$b$};
    \draw[arrow] (A2) -- (I) node[midway, above] {};
    \draw[arrow] (I) -- (B2) node[midway, below] {$d$};
\draw[arrow] (M3) -- (C2) node[midway, above] {$c$};
\draw[arrow] (M3) -- (A2) node[midway, above] {$a$};
    \draw[arrow] (I) -- (F) node[midway, above] {$d$};
    \draw[arrow] (C2) -- (I) node[midway, above] {$a$};
    \draw[arrow] (I) -- (D2) node[midway, below] {$b$};

    \draw[arrow] (E) -- (A1) node[midway,  left] {$c$};
    \draw[arrow] (A2) -- (A1) node[midway, above] {$b$};
    \draw[arrow] (E) -- (B1) node[midway, left] {$d$};
    \draw[arrow] (B2) -- (B1) node[midway, below] {$b$};

    \draw[arrow] (C1) -- (F) node[midway,  right] {$a$};
    \draw[arrow] (C2) -- (C1) node[midway, above] {$d$};
    \draw[arrow] (F) -- (D1) node[midway, right] {$b$};
    \draw[arrow] (D2) -- (D1) node[midway, below] {$d$};

\path let \p1=(F), \p2=(C1), \p3=(C2), \p4=(I) in
  coordinate (urTrapCenter) at ({(\x1+\x2+\x3+\x4)/4},{(\y1+\y2+\y3+\y4)/4});

  \draw[->, very thick, orange!80!black]
    (M3) -- (midMAtwo) -- (upperCenter) -- (midA2I) -- (I) -- (midIE) -- (trapCenter) -- (midNE)-- (NEW);
 \begin{scope}[transform canvas={xshift=3pt}]
    \draw[->, very thick, green] (M3) --  ($(M3)!0.5!(A2)$)
 -- (upperCenter) -- ($(C2)!0.5!(I)$)
 -- (urTrapCenter) -- ($(I)!0.5!(F)$)
 -- (trapCenter) -- ([xshift=2pt,yshift=1pt]midNE)
-- ([yshift=2pt]NEW);
  \end{scope}
  \begin{scope}[transform canvas={xshift=-3pt}]
    \draw[-, very thick, blue]
      (M3) -- (midMAtwo) -- (upperCenter) -- ([xshift=-2pt]midA2I);
  \end{scope}
  \begin{scope}[transform canvas={xshift=-5pt}]
    \draw[->, very thick, blue]
       ([xshift=1pt]midA2I) -- (ulTrapCenter) -- (midIE) -- (trapCenter) -- ([yshift=-1pt]midNE)-- ([yshift=-1pt]NEW);
  \end{scope}
  \node[font=\footnotesize] at (.2,1.2) {\textcolor{orange!80!black}{$\beta$}};

     \path let \p1=(I), \p2=(E), \p3=(B1), \p4=(B2) in
      coordinate (dlTrapCenter) at ({(\x1+\x2+\x3+\x4)/4},{(\y1+\y2+\y3+\y4)/4});

      \end{scope}

    \end{tikzpicture}

    \caption{Two dimensional HDAs $\mathcal{X}$ on the left and $\mathcal{Y}$ on the right that are shp-bisimilar but not qhp-bisimilar (start vertex marked by an incoming arrow).}
    \label{fig:projected HDA}
\end{figure}

 \smallskip
\noindent\emph{Why $\mathcal X \bis{shp} \mathcal Y$.}
We now exhibit an shp-bisimulation $\mathfrak S\subseteq\Path_X\times\Path_Y$.
Each path in $\mathcal Y$ that does not pass through the middle state is related to the unique path in $\mathcal X$ that has the same ST-trace. For paths that do go through the middle state,
if from there they enter a wing of $\mathcal Y$, the matching path of $\mathcal X$ enters the same wing of $\mathcal X$; and if the first move after the middle state doesn't enter a wing, neither does the corresponding path of $\mathcal X$. The continuations of these paths after that move are completely determined by their ST-traces. It is trivial to check that  $\mathfrak S$ satisfies all clauses of a shp-bisimulation. In fact, it is even an $\{\subsu,\unsubsu,\hookleftarrow\}$-bisimulation.

Note that $\mathfrak S$ relates $\alpha_1$ with the path $\beta$ considered above. Yet, the counterexample cannot be completed, because an shp-bisimulation lacks the similarity clause.
\end{proof}

\begin{theorem}\rm\label{xST}
  ${\bis{reST}} \subsetneq {\bis{ST}}$ and ${\bis{hhp}} \subsetneq {\bis{hp}}$.
  In fact, even ${\bis{hp}} \not\subseteq {\bis{\{\simeq,\hookleftarrow\}}} = {\bis{reST}}$
  and ${\bis{hshp}}={\bis{\{\subsu,\unsubsu,\hookleftarrow\}}} \not\subseteq {\bis{reST}}$.

\end{theorem}
\begin{proof}
The inclusions are immediate from the definitions.
Strictness follows from the absorption-law example in Fig.~\ref{fig:absorption law}.
This example is used in \cite{van2001refinement,BC14} to separate the hereditary and non-hereditary variants, and the same argument applies verbatim here.
To distinguish these HDAs one needs both $\simeq$ and $\hookleftarrow$.
For completeness, we will give an independent separation argument later in Section~\ref{sec: bis logic}, exhibiting a distinguishing $\lan{reST}$-formula for the two systems.
\end{proof}

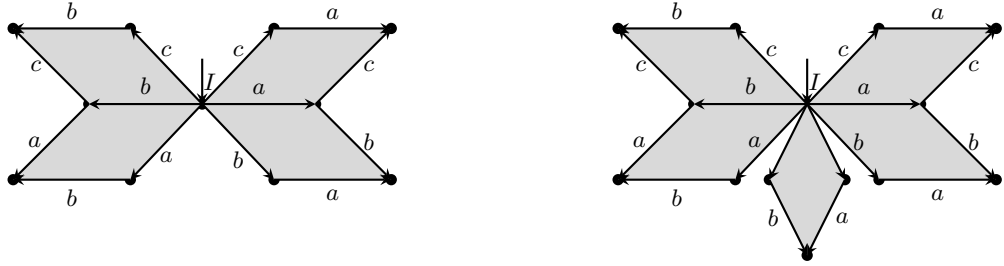
\begin{figure}[ht]
    \centering
    \begin{tikzpicture}
        \begin{scope}[shift={(-4,0)}] 
         \tikzset{
        point/.style={circle, fill, inner sep=1.5pt},
        dot/.style={circle, fill, inner sep=0.8pt},
        arrow/.style={->, >=stealth, line width=0.8pt},
        label/.style={font=\small},
    }

    \coordinate (I) at (0, 0);
    \coordinate (E) at (-1.5, 0);
    \coordinate (F) at (1.5, 0); 

    \draw[arrow] (0,.6) -- (I);

    \coordinate (A1) at (-2.5, 1);
    \coordinate (A2) at (-0.5-.45, 1);
    \coordinate (B1) at (-2.5, -1);
    \coordinate (B2) at (-0.5-.45, -1);

    \coordinate (C1) at (2.5, 1);
    \coordinate (C2) at (0.5+.45, 1);
    \coordinate (D1) at (2.5, -1);
    \coordinate (D2) at (0.5+.45, -1);

    \foreach \p in {(I), (E), (F), (A1), (A2), (B1), (B2), (C1), (C2), (D1), (D2)}{
        \node[point] at \p {};
    }
  
    \node[dot] at (-1.5, 0.5) {};
    \node[dot] at (-1.5, -0.5) {};
    \node[dot] at (1.5, 0.5) {};
    \node[dot] at (1.5, -0.5) {};
 \path[fill=black!15]  (E) to (A1) to (A2) to (I) to (E);
  \path[fill=black!15]  (E) to (B1) to (B2) to (I) to (E);
  \path[fill=black!15] (F) to (C1) to (C2) to (I) to (F);
  \path[fill=black!15] (F) to (D1) to (D2) to (I) to (F);
    \draw[arrow] (I) -- (E) node[midway, above] {$b$};
    \draw[arrow] (I) -- (A2) node[midway, above] {$c$};
    \draw[arrow] (I) -- (B2) node[midway, below] {$a$};

    \draw[arrow] (I) -- (F) node[midway, above] {$a$};
    \draw[arrow] (I) -- (C2) node[midway, above] {$c$};
    \draw[arrow] (I) -- (D2) node[midway, below] {$b$};

    \draw[arrow] (E) -- (A1) node[midway,  left] {$c$};
    \draw[arrow] (A2) -- (A1) node[midway, above] {$b$};
    \draw[arrow] (E) -- (B1) node[midway, left] {$a$};
    \draw[arrow] (B2) -- (B1) node[midway, below] {$b$};

    \draw[arrow] (F) -- (C1) node[midway,  right] {$c$};
    \draw[arrow] (C2) -- (C1) node[midway, above] {$a$};
    \draw[arrow] (F) -- (D1) node[midway, right] {$b$};
    \draw[arrow] (D2) -- (D1) node[midway, below] {$a$};

    \node[label] at (0.1, 0.3) {$I$};
      \end{scope}

        \begin{scope}[shift={(4,0)}] 
                \tikzset{
        point/.style={circle, fill, inner sep=1.5pt},
        dot/.style={circle, fill, inner sep=0.8pt},
        arrow/.style={->, >=stealth, line width=0.8pt},
        label/.style={font=\small},
    }

    \coordinate (I) at (0, 0);
    \coordinate (E) at (-1.5, 0);
    \coordinate (F) at (1.5, 0); 

    \draw[arrow] (0,.6) -- (I);

    \coordinate (A1) at (-2.5, 1);
    \coordinate (A2) at (-0.5-.45, 1);
    \coordinate (B1) at (-2.5, -1);
    \coordinate (B2) at (-0.5-.45, -1);

    \coordinate (C1) at (2.5, 1);
    \coordinate (C2) at (0.5+.45, 1);
    \coordinate (D1) at (2.5, -1);
    \coordinate (D2) at (0.5+.45, -1);

    \coordinate (M1) at (.5, -1);
    \coordinate (M2) at (-.5, -1);
    \coordinate (M3) at (0, -2);
    
    \foreach \p in {(I), (E), (F), (A1), (A2), (B1), (B2), (C1), (C2), (D1), (D2),(M1),(M2),(M3)}{
        \node[point] at \p {};
    }
    \node[dot] at (-1.5, 0.5) {};
    \node[dot] at (-1.5, -0.5) {};
    \node[dot] at (1.5, 0.5) {};
    \node[dot] at (1.5, -0.5) {};
 \path[fill=black!15]  (E) to (A1) to (A2) to (I) to (E);
  \path[fill=black!15]  (E) to (B1) to (B2) to (I) to (E);
  \path[fill=black!15] (F) to (C1) to (C2) to (I) to (F);
  \path[fill=black!15] (F) to (D1) to (D2) to (I) to (F);
  \path[fill=black!15] (M3) to (M1) to (I) to (M2) to (M3);
  
    \draw[arrow] (I) -- (M1) node[midway,below right] {};
    \draw[arrow] (I) -- (M2) node[midway,below left] {};

    \draw[arrow] (M1) -- (M3) node[midway,right] {$a$};
    \draw[arrow] (M2) -- (M3) node[midway,left] {$b$};

    \draw[arrow] (I) -- (E) node[midway, above] {$b$};
    \draw[arrow] (I) -- (A2) node[midway, above] {$c$};
    \draw[arrow] (I) -- (B2) node[midway, left] {$a$};

    \draw[arrow] (I) -- (F) node[midway, above] {$a$};
    \draw[arrow] (I) -- (C2) node[midway, above] {$c$};
    \draw[arrow] (I) -- (D2) node[midway, right] {$b$};

    \draw[arrow] (E) -- (A1) node[midway,  left] {$c$};
    \draw[arrow] (A2) -- (A1) node[midway, above] {$b$};
    \draw[arrow] (E) -- (B1) node[midway, left] {$a$};
    \draw[arrow] (B2) -- (B1) node[midway, below] {$b$};

    \draw[arrow] (F) -- (C1) node[midway,  right] {$c$};
    \draw[arrow] (C2) -- (C1) node[midway, above] {$a$};
    \draw[arrow] (F) -- (D1) node[midway, right] {$b$};
    \draw[arrow] (D2) -- (D1) node[midway, below] {$a$};
    \node[label] at (0.1, 0.3) {$I$};
    \end{scope}
    \end{tikzpicture}

    \caption{The absorption law from \cite{BC14,van2001refinement} depicting two systems 
$E_1 = \big(a \mid (b + c)\big) + \big(b \mid (a + c)\big) 
 \text{ on the left and } 
E_2 = \big(a \mid (b + c)\big) + a \mid b + \big(b \mid (a + c)\big)$ on the right.
These HDA are hp-bisimilar as well as $\{\subsu,\unsubsu,\hookleftarrow\}$-bisimilar, but not reST-bisimilar.}
    \label{fig:absorption law}
\end{figure}

\begin{figure}[hb]
    \centering
    \begin{tikzpicture}[scale=1.58]
\def\colC{red}
\def\colD{orange}
\def\colA{green!80!black}
\def\colB{blue}
\def\colE{violet}
\def\colF{teal}
\def\colG{olive}
\def\sh{2.3}
\def\sk{4.5}
\def\sl{6.7}
    \node at(0.8,1.7) {$\mathcal{A}$};
        \path[fill=black!8] (0,0) to (1,0) to (1.6,0.4) to (1.6,1.4) to (0.6,1.4) to (0,1) to (0,0);
        \node[state] (000) at(0,0) {};
        \node[state] (100) at(1,0) {};
        \node[below right] at (100) {$x_1$};
        \node[state] (010) at(0,1) {};
        \node[state] (110) at(1,1) {};
        \node[state] (001) at(0.6,0.4) {};
        \node[state] (101) at(1.6,0.4) {};
        \node[below right] at (101) {$y$};
        \node[state] (011) at(0.6,1.4) {};
        \node[state] (111) at(1.6,1.4) {};
        \node[right] at (111) {$z$};
        \draw[->] (000)--(001);
        \draw[->] (010) edge node[above left] {$a$} (011);
        \draw[thick,->,\colB] (100) edge node[below right] {$a'_1$}(101);
        \draw[->] (110)--(111);
        \draw[->] (000) edge node[left] {$b$} (010);
        \draw[->] (001)--(011);
        \draw[->] (100)--(110);
        \draw[thick,->,\colE] (101) edge node[right] {$b''$} (111);
        \draw[thick,->,\colA] (000) edge node[below] {$a_1$} (100);
        \draw[->] (001)--(101);
        \draw[->] (010)--(110);
        \draw[->] (011)--(111);
        \draw[->] (-0.2,0)--(000);
    \node at(\sh+0.8,1.7) {$\mathcal{B}_1$};
        \path[fill=black!15] (\sh,0) to (\sh+1,0) to (\sh+1,1) to (\sh,1) to (\sh,0);
        \node[state] (00a) at(\sh,0) {};
        \node[state] (10a) at(\sh+1,0) {};
        \node[below right] at (10a) {$x_2$};
        \node[state] (01a) at(\sh,1) {};
        \node[above left] at (01a) {$x_1$};
        \node[state] (11a) at(\sh+1,1) {};
        \node[below right] at (11a) {$y$};
        \node[state] (a) at(\sh+1.6,1.4) {};
        \node[right] at (a) {$z$};
        \draw[thick,->,\colA] (00a) edge node[left] {$a_1$} (01a);
        \draw[thick,->,\colD] (10a) edge node[right] {$a'_2$} (11a);
        \draw[thick,->,\colC] (00a) edge node[below] {$a_2$} (10a);
        \draw[thick,->,\colB] (01a) edge node[above] {$a'_1$} (11a);
        \draw[thick,->,\colE] (11a) edge node[above left, xshift=5pt] {$b''$} (a);
        \draw[->] (\sh-0.2,0)--(00a);
        \node at (\sh+0.5,0.5) {$q_1$};
    \node at(\sk+0.8,1.7) {$\mathcal{B}_2$};
        \path[fill=black!15] (\sk,0) to (\sk+1,0) to (\sk+1,1) to (\sk,1) to (\sk,0);
        \node[state] (00b) at(\sk,0) {};
        \node[state] (10b) at(\sk+1,0) {};
        \node[below right] at (10b) {$x_3$};
        \node[state] (01b) at(\sk,1) {};
        \node[above left] at (01b) {$x_2$};
        \node[state] (11b) at(\sk+1,1) {};
        \node[below right] at (11b) {$y$};
        \node[state] (b) at(\sk+1.6,1.4) {};
        \node[right] at (b) {$z$};
        \draw[thick,->,\colC] (00b) edge node[left] {$a_2$} (01b);
        \draw[thick,->,\colG] (10b) edge node[right] {$a'_3$} (11b);
        \draw[thick,->,\colF] (00b) edge node[below] {$a_3$} (10b);
        \draw[thick,->,\colD] (01b) edge node[above] {$a'_2$} (11b);
        \draw[thick,->,\colE] (11b) edge node[above left, xshift=5pt] {$b''$} (b);
        \draw[->] (\sk-0.2,0)--(00b);
        \node at (\sk+0.5,0.5) {$q_2$};
    \node at(\sl+0.8,1.7) {$\mathcal{B}_3$};
        \path[fill=black!15] (\sl,0) to (\sl+1,0) to (\sl+1,1) to (\sl,1) to (\sl,0);
        \node[state] (00b) at(\sl,0) {};
        \node[state] (10b) at(\sl+1,0) {};
        \node[below right] at (10b) {$x_4$};
        \node[state] (01b) at(\sl,1) {};
        \node[above left] at (01b) {$x_3$};
        \node[state] (11b) at(\sl+1,1) {};
        \node[below right] at (11b) {$y$};
        \node[state] (b) at(\sl+1.6,1.4) {};
        \node[right] at (b) {$z$};
        \draw[thick,->,\colF] (00b) edge node[left] {$a_3$} (01b);
        \draw[->] (10b) edge node[right] {$a'_4$} (11b);
        \draw[->] (00b) edge node[below] {$a_4$} (10b);
        \draw[thick,->,\colG] (01b) edge node[above] {$a'_3$} (11b);
        \draw[thick,->,\colE] (11b) edge node[above left, xshift=5pt] {$b''$} (b);
        \draw[->] (\sl-0.2,0)--(00b);
        \node at (\sl+0.5,0.5) {$q_3$};
        \node at (8.3,0.5) {$\dots$};
\end{tikzpicture}
    \caption{The HDA $\mathcal{X}_n$ is obtained from the disjoint union of HDAs $\mathcal{A}$, $\mathcal{B}_1, \mathcal{B}_2,\dotsc, \mathfrak{B}_n$ by identification of edges and vertices having the same names. All edges $a_i$ and $a_i'$ are labelled with $a$, and $b''$ is labelled by $b$.}
    \label{fig:hp-not-qhp2}
\end{figure}
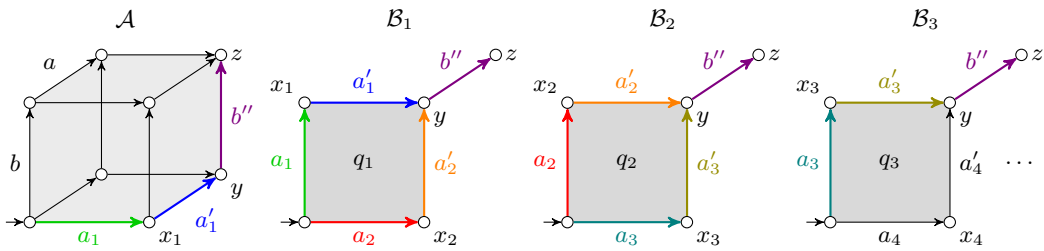

In order to show the difference between $\bis{qhp}$- and $\bis{hp}$-bisimilarity,
let $\mathcal{X}_n$ for $n\geq 1$ be the HDA that is obtained by gluing HDAs $\mathcal{A}$ and $\mathcal{B}_1,\dotsc,\mathcal{B}_n$ displayed in Fig.~\ref{fig:hp-not-qhp2}.
For the following observations, Figure \ref{fig:snake-paths} is helpful.

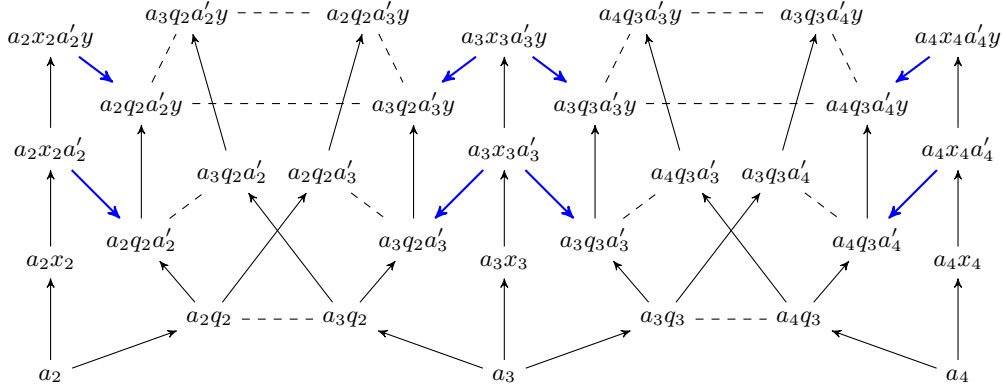
\begin{figure}
    \centering
    \begin{tikzpicture}[xscale=3, yscale=1.5]
        \node (00) at (0,0) {$a_3$};
        \node (01) at (0,1) {$a_3x_3$};
        \node (02) at (0,2) {$a_3x_3a'_3$};
        \node (03) at (0,3) {$a_3x_3a'_3y$};
        \draw[->] (00)--(01);
        \draw[->] (01)--(02);
        \draw[->] (02)--(03);
        \node (20) at (2,0) {$a_4$};
        \node (21) at (2,1) {$a_4x_4$};
        \node (22) at (2,2) {$a_4x_4a'_4$};
        \node (23) at (2,3) {$a_4x_4a'_4y$};
        \draw[->] (20)--(21);
        \draw[->] (21)--(22);
        \draw[->] (22)--(23);
        \node(1a) at (0.7, 0.5) {$a_3q_3$};
        \node(2a) at (1.3, 0.5) {$a_4q_3$};
        \draw[->] (00)--(1a);
        \draw[->] (20)--(2a);
        \draw[-,dashed] (1a)--(2a);
        \node(1b) at (0.4, 1.2) {$a_3q_3a'_3$};
        \node(2b) at (0.8, 1.8) {$a_4q_3a'_3$};
        \node(3b) at (1.2, 1.8) {$a_3q_3a'_4$};
        \node(4b) at (1.6, 1.2) {$a_4q_3a'_4$};
        \draw[->] (1a)--(1b);
        \draw[->] (1a)--(3b);
        \draw[->] (2a)--(2b);
        \draw[->] (2a)--(4b);
        \draw[-,dashed] (1b)--(2b);
        \draw[-,dashed] (3b)--(4b);
        \draw[->,thick,blue] (02)--(1b);
        \draw[->,thick,blue] (22)--(4b);
        \node(1c) at (0.4, 2.4) {$a_3q_3a'_3y$};
        \node(2c) at (0.6, 3.2) {$a_4q_3a'_3y$};
        \node(3c) at (1.4, 3.2) {$a_3q_3a'_4y$};
        \node(4c) at (1.6, 2.4) {$a_4q_3a'_4y$};
        \draw[->] (1b)--(1c);
        \draw[->] (2b)--(2c);
        \draw[->] (3b)--(3c);
        \draw[->] (4b)--(4c);
        \draw[-,dashed] (1c)--(2c);
        \draw[-,dashed] (3c)--(4c);
        \draw[-,dashed] (1c)--(4c);
        \draw[-,dashed] (2c)--(3c);
        \draw[->,thick,blue] (03)--(1c);
        \draw[->,thick,blue] (23)--(4c);
        \node (-0) at (-2,0) {$a_2$};
        \node (-1) at (-2,1) {$a_2x_2$};
        \node (-2) at (-2,2) {$a_2x_2a'_2$};
        \node (-3) at (-2,3) {$a_2x_2a'_2y$};
        \draw[->] (-0)--(-1);
        \draw[->] (-1)--(-2);
        \draw[->] (-2)--(-3);
        \node(1q) at (-1.3, 0.5) {$a_2q_2$};
        \node(2q) at (-0.7, 0.5) {$a_3q_2$};
        \draw[->] (-0)--(1q);
        \draw[->] (00)--(2q);
        \draw[-,dashed] (1q)--(2q);
        \node(1r) at (-1.6, 1.2) {$a_2q_2a'_2$};
        \node(2r) at (-1.2, 1.8) {$a_3q_2a'_2$};
        \node(3r) at (-0.8, 1.8) {$a_2q_2a'_3$};
        \node(4r) at (-0.4, 1.2) {$a_3q_2a'_3$};
        \draw[->] (1q)--(1r);
        \draw[->] (1q)--(3r);
        \draw[->] (2q)--(2r);
        \draw[->] (2q)--(4r);
        \draw[-,dashed] (1r)--(2r);
        \draw[-,dashed] (3r)--(4r);
        \draw[->,thick,blue] (-2)--(1r);
        \draw[->,thick,blue] (02)--(4r);
        \node(1s) at (-1.6, 2.4) {$a_2q_2a'_2y$};
        \node(2s) at (-1.4, 3.2) {$a_3q_2a'_2y$};
        \node(3s) at (-0.6, 3.2) {$a_2q_2a'_3y$};
        \node(4s) at (-0.4, 2.4) {$a_3q_2a'_3y$};
        \draw[->] (1r)--(1s);
        \draw[->] (2r)--(2s);
        \draw[->] (3r)--(3s);
        \draw[->] (4r)--(4s);
        \draw[-,dashed] (1s)--(2s);
        \draw[-,dashed] (3s)--(4s);
        \draw[-,dashed] (1s)--(4s);
        \draw[-,dashed] (2s)--(3s);
        \draw[->,thick,blue] (-3)--(1s);
        \draw[->,thick,blue] (03)--(4s);
\end{tikzpicture}
    \caption{A part of the diagram of paths in $\mathcal{X}_n$. Regular arrows denote path extensions, dashed lines denote similarity and blue arrows denote subsumption. Various $\mathcal{X}_n$'s are bisimilar if and only if the diagram can be traversed sideways using arrows having the given set of types.}
    \label{fig:snake-paths}
\end{figure}

\begin{observation}\rm
    $\Path_{\mathcal{X}_n}=\Path_{\mathcal{X}_1}\cup \bigcup_{r\geq 2}\Path_{\mathcal{B}_r}$.
\end{observation}

\begin{observation}\rm
\label{lem:snake-locality}
    Let $\alpha\in \Path_{\mathcal{X}_n}$ be a path that can be obtained from $(i,a_r)$ by applying the closure operations
    \begin{enumerate}
        \item $\ininc, \subsu^\ell, \simeq$,
        \item $\ininc, \sqsupseteq^\ell, \simeq$, or    
        \item $\ininc, \subsu^\ell, \unsubsu^\ell, \hookleftarrow$.
    \end{enumerate}
    Then $\alpha \in \Path_{\mathcal{X}_r}$ for all $r$, and 
    $\alpha \in \Path_{\mathcal{B}_{r-1}}\cup \Path_{\mathcal{B}_{r}}$
    for $r\geq 3$.
\end{observation}

\begin{proposition}\rm
    The HDAs $\mathcal{X}_m$ and $\mathcal{X}_n$ for $3\leq m<n$ in Fig.\@ \ref{fig:hp-not-qhp2} are:
    \begin{enumerate}
    \item not $\{{\subsu}, {\sqsupseteq}, {\simeq}\}$-bisimilar (hp).
    \item not $\{{\hookleftarrow}, {\simeq}\}$-bisimilar (reST).
    \item $\{{\subsu}, {\simeq}\}$-bisimilar (qhp).
    \item $\{{\sqsupseteq}, {\simeq}\}$-bisimilar ($=$ ST).
    \item $\{{\subsu}, {\sqsupseteq}, {\hookleftarrow}\}$-bisimilar ($=$ hshp).
    \end{enumerate}
\end{proposition}
\begin{proof}
Note that $\mathcal{X}_m\subseteq \mathcal{X}_n$.

    1. 
    Let $\mathfrak{R}$ be a bisimulation between $\mathcal{X}_m$ and $\mathcal{X}_n$. 
    Let $\omega_r$ be the path $(i,a_r,x_r,a_r',y,b'',z)$. Note that $\omega_1$ cannot be related to $\omega_r$ for $r\geq 2$: $\omega_1$ can be subsumed by a path having trace $a^+a^1a^+b^+b^4a^3$ while $\omega_r$ cannot.

    Consider a sequence of paths:
      $$  \omega_r \subsu^2 \alpha \simeq^{1,3} \alpha' \sqsupseteq^2 \alpha''. $$
    Then $\alpha''=\omega_{r-1}$ or $\alpha_{r+1}$. As a consequence, $\omega_2$ cannot be related to any path $\omega_r$ for $r\geq 3$. Following this argument we obtain that $\mathcal{X}_m$ and $\mathcal{X}_n$ can be bisimilar only for $m=n$.

    2. The argument is similar. Let $\eta_r=(i,a_r)$. Again, $\eta_1$ cannot be bisimilar to $\eta_r$ for $r>1$: $\eta_1$ can be extended to a path with ST-trace $a^+b^+$ while $\eta_r$ cannot. Furthermore, the sequence
    $$ \eta_r \ininc \beta \simeq^1 \beta' \hookleftarrow \beta''$$
    with $\ST(\beta)=\ST(\beta')=a^+a^+$ and $\ST(\beta'')=a^+$ can only produce $\beta''=\eta^s$ for $s\in\{r-1,r+1\}$. Thus, we show by induction that $\eta_r\mathfrak{R}\eta_s \implies r=s$.
    
    3. Let $\mathscr{P}_{\geq 3}^n\subseteq \Path_{\mathcal{X}_n}$ be the set of paths that can be obtained from $(i,a_r)$, $r\geq 3$, by applying the closure operations $\ininc$, $\subsu$ and $\simeq$.
    Namely, $\mathscr{P}_{\geq 3}^n$ consists of paths $(i,a_r)$ for $r\geq 3$ and paths of length $\geq 2$ that cross $x_r$ for $r\geq 3$ or $q_r$ for $r\geq 2$. Observation~\ref{lem:snake-locality} implies that $\mathscr{P}_{\geq 3}^n\subseteq \bigcup_{r\geq 2}\Path_{\mathcal{B}_r}$.
    Define a relation $\mathfrak{R}\subseteq \Path_{\mathcal{X}_m}\times \Path_{\mathcal{X}_n}$ by
       $$ \alpha\mathfrak{R}\beta \iff \alpha=\beta \;\lor \; (\alpha\in \mathscr{P}_{\geq 3}^m\;\land\; \beta\in \mathscr{P}_{\geq 3}^n\;\land\; \ST(\alpha)=\ST(\beta)).$$
    Another consequence of Observation~\ref{lem:snake-locality} is that $\Path_{\mathcal{X}_n}\setminus \mathscr{P}_{\geq 3}^n=\Path_{\mathcal{X}_m}\setminus \mathscr{P}_{\geq 3}^m$; denote it $\mathfrak{A}$. Paths in $\mathfrak{A}$ are never related to paths outside $\mathfrak{A}$.
    
    Assume that $\alpha\mathfrak{R}\beta$ and $\beta \ltimes \beta'$ for $\ltimes\in\{{\ininc}, {\subsu}, {\simeq}\}$. If $\alpha\in\mathfrak{A}$, then $\beta=\alpha$ and we have $\alpha\ltimes \beta'$ (and $\beta'\mathfrak{R}\beta'$).
    
    If $\alpha\in \mathscr{P}_{\geq 3}^m$, then $\beta'\in \mathscr{P}_{\geq 3}^n$. This implies that $\beta,\beta'\in\Path_{\mathcal{B}_s}$ for some $s\geq 2$. Similarly, $\alpha\in \Path_{\mathcal{B}_r}$ for $r\geq 2$. Moreover, there exists a map $f:\mathcal{B}_s\to\mathcal{B}_r$, which is either
\begin{align*}
    (q_s,a_s,a_{s+1},a'_s,a'_{s+1},x_s,x_{s+1})&\mapsto
    (q_r,a_r,a_{r+1},a'_r,a'_{r+1},x_r,x_{r+1})
    \\
    (q_s,a_s,a_{s+1},a'_s,a'_{s+1},x_s,x_{s+1})&\mapsto
    (q_r,a_{r+1},a_{r},a'_{r+1},a'_{r},x_{r+1},x_{r})
\end{align*}
    (an isomorphism with a possible swap), and the identity on $i_{\mathcal X_n}$, $b''$ and $z$, such that $f(\beta)=\alpha$. We put $\alpha'=f(\beta')$. Clearly $\alpha\ltimes \alpha'$, $\ST(\alpha')=\ST(\beta')$, and $\alpha',\beta'\not\in \mathfrak{A}$ so $\alpha'\mathfrak R \beta'$.

    4., 5. A similar argument applies.
\end{proof}

\begin{corollary}\rm
${\bis{hp}} \subsetneq {\bis{qhp}}$.
\end{corollary}

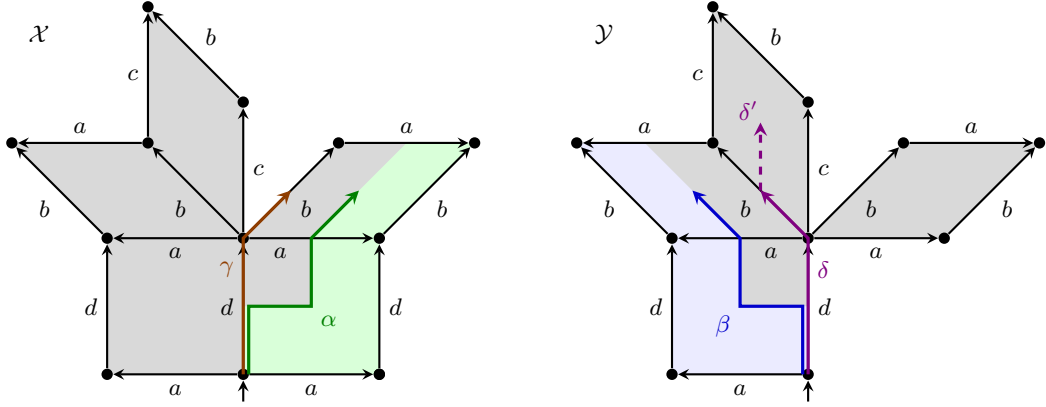
\begin{figure}
    \centering
    \begin{tikzpicture}[scale=0.9]
      \tikzset{
        point/.style={circle, fill, inner sep=1.5pt},
        dot/.style={circle, fill, inner sep=0.8pt},
        arrow/.style={->, >=stealth, line width=0.8pt},
        label/.style={font=\small},
    }
    \begin{scope}[shift={(-4,0)}] 
    \node at (-1,5) {$\mathcal{X}$};
    \coordinate (00) at (0, 0);
    \coordinate (10) at (2, 0);
    \coordinate (20) at (4, 0);
    \coordinate (01) at (0, 2);
    \coordinate (11) at (2, 2);
    \coordinate (21) at (4, 2);
    \coordinate (a) at (-1.4, 3.4);
    \coordinate (b) at (0.6, 3.4);
    \coordinate (c) at (3.4, 3.4);
    \coordinate (d) at (5.4, 3.4);
    \coordinate (e) at (2, 4);
    \coordinate (f) at (0.6, 5.4);
    \path[fill=black!15]  (00) to (20) to (21) to (d) to (c) to (11) to (e) to (f) to (b) to (a) to (01) to (00);
    \path[fill=green!15] (2,0) to (2,1) to (3,1) to (3,2) to (4.4,3.4) to (5.4,3.4) to (4,2) to (4,0) to (2,0);
    \node[point] (p00) at (00) {};
    \node[point] (p10) at (10) {};
    \node[point] (p20) at (20) {};
    \node[point] (p01) at (01) {};
    \node[point] (p11) at (11) {};
    \node[point] (p21) at (21) {};
    \node[point] (pa) at (a) {};
    \node[point] (pb) at (b) {};
    \node[point] (pc) at (c) {};
    \node[point] (pd) at (d) {};
    \node[point] (pe) at (e) {};
    \node[point] (pf) at (f) {};
    \draw[arrow] (p10) -- (p00) node[midway,below] {$a$};
    \draw[arrow] (p10) -- (p20) node[midway,below] {$a$};
    \draw[arrow] (p11) -- (p01) node[midway,below] {$a$};
    \draw[arrow] (p11) -- (p21) node[near start,below] {$a$};
    \draw[arrow] (pb) -- (pa) node[midway,above] {$a$};
    \draw[arrow] (pc) -- (pd) node[midway,above] {$a$};
    \draw[arrow] (p00) -- (p01) node[midway,left] {$d$};
    \draw[arrow] (p10) -- (p11) node[midway,left] {$d$};
    \draw[arrow] (p20) -- (p21) node[midway,right] {$d$};
    \draw[arrow] (p01) -- (pa) node[midway,below left] {$b$};
    \draw[arrow] (p11) -- (pb) node[midway,below left] {$b$};
    \draw[arrow] (p11) -- (pc) node[midway,below right] {$b$};
    \draw[arrow] (p21) -- (pd) node[midway,below right] {$b$};
    \draw[arrow] (pe) -- (pf) node[midway,above right] {$b$};
    \draw[arrow] (p11) -- (pe) node[midway,right] {$c$};
    \draw[arrow] (pb) -- (pf) node[midway,left] {$c$};
    \draw[arrow] (2,-0.4) -- (p10);
    \draw[arrow,green!50!black,very thick] (2.08,0)--(2.08,1)--(3,1)--(3,2)--(3.7,2.7);
    \node[below right,green!50!black] at (3,1) {$\alpha$};
    \draw[arrow,green!30!red!80!black,very thick] (2.00,0)--(2.00,2)--(2.7,2.7);
    \node[below left,green!30!red!80!black] at (2,1.8) {$\gamma$};
    \end{scope}
    \begin{scope}[shift={(4.3,0)}] 
    \node at (-1,5) {$\mathcal{Y}$};
    \coordinate (00) at (0, 0);
    \coordinate (10) at (2, 0);
    \coordinate (01) at (0, 2);
    \coordinate (11) at (2, 2);
    \coordinate (21) at (4, 2);
    \coordinate (a) at (-1.4, 3.4);
    \coordinate (b) at (0.6, 3.4);
    \coordinate (c) at (3.4, 3.4);
    \coordinate (d) at (5.4, 3.4);
    \coordinate (e) at (2, 4);
    \coordinate (f) at (0.6, 5.4);
    \path[fill=black!15]  (00) to (10) to (11) to (21) to (d) to (c) to (11) to (e) to (f) to (b) to (a) to (01) to (00);
    \path[fill=blue!08] (2,0) to (2,1) to (1,1) to (1,2) to (-0.4,3.4) to (-1.4,3.4) to (0,2) to (0,0) to (2,0);
    \node[point] (p00) at (00) {};
    \node[point] (p10) at (10) {};
    \node[point] (p20) at (20) {};
    \node[point] (p01) at (01) {};
    \node[point] (p11) at (11) {};
    \node[point] (p21) at (21) {};
    \node[point] (pa) at (a) {};
    \node[point] (pb) at (b) {};
    \node[point] (pc) at (c) {};
    \node[point] (pd) at (d) {};
    \node[point] (pe) at (e) {};
    \node[point] (pf) at (f) {};
    \draw[arrow] (p10) -- (p00) node[midway,below] {$a$};
    \draw[arrow] (p11) -- (p01) node[near start,below] {$a$};
    \draw[arrow] (p11) -- (p21) node[midway,below] {$a$};
    \draw[arrow] (pb) -- (pa) node[midway,above] {$a$};
    \draw[arrow] (pc) -- (pd) node[midway,above] {$a$};
    \draw[arrow] (p00) -- (p01) node[midway,left] {$d$};
    \draw[arrow] (p10) -- (p11) node[midway,right] {$d$};
    \draw[arrow] (p01) -- (pa) node[midway,below left] {$b$};
    \draw[arrow] (p11) -- (pb) node[midway,below left] {$b$};
    \draw[arrow] (p11) -- (pc) node[midway,below right] {$b$};
    \draw[arrow] (p21) -- (pd) node[midway,below right] {$b$};
    \draw[arrow] (pe) -- (pf) node[midway,above right] {$b$};
    \draw[arrow] (p11) -- (pe) node[midway,right] {$c$};
    \draw[arrow] (pb) -- (pf) node[midway,left] {$c$};
    \draw[arrow] (2,-0.4) -- (p10);
    \draw[arrow,blue!80!black,very thick] (1.92,0)--(1.92,1)--(1,1)--(1,2)--(0.3,2.7);
    \node[below left,blue!80!black] at (1,1) {$\beta$};
    \draw[arrow,blue!50!red,very thick] (2.00,0)--(2.00,2)--(1.3,2.7);
    \node[below right,blue!50!red] at (2,1.8) {$\delta$};
    \draw[arrow,blue!50!red,very thick,dashed] (1.3,2.7)--(1.3,3.7);
    \node[above left,blue!50!red] at (1.4,3.6) {$\delta'$};
    \end{scope}
    \end{tikzpicture}
    \caption{An example of HDAs that are reST-bisimilar but not ureST-bisimilar.}
    \label{fig:ur-not-ure}
\end{figure}

\begin{theorem}
    ${\bis{ureST}} \subsetneq {\bis{reST}}$.
\end{theorem}
\begin{proof}
    Let $\mathcal{X}$, $\mathcal{Y}$ be the HDAs from Figure \ref{fig:ur-not-ure}. Note that $\mathcal{Y}\subseteq \mathcal{X}$. Define a relation $\mathfrak{R}\subseteq \Path_{\mathcal{X}}\times \Path_{\mathcal{Y}}$ that contains all pairs $(\alpha,\alpha)$ for $\alpha\in \Path_{\mathcal{Y}}$ and all pairs $(\alpha,\beta)$, where $\alpha\in \Path_{\mathcal{X}}\setminus\Path_{\mathcal{Y}}$ (green area) and $\beta$ is the reflection of $\alpha$ along the vertical axis (blue area). Restrictions, extensions and similarities of paths in the green area stay in the green area (or in $\Path_{\mathcal{Y}}$) and then can be matched with similar operations performed in the blue area.
    Likewise, restrictions, extensions and similarities of paths of $\mathcal{X}$ that lay in $\mathcal{Y}$ stay within $\mathcal{Y}$ (or in the green area) and can be matched by themselves. Thus $\mathfrak{R}$ is a reST-bisimulation.

    Assume that $\mathcal{X}$ and $\mathcal{Y}$ are ureST-bisimilar. Then the path $\alpha\in \Path_{\mathcal{X}}$ must be matched by $\beta\in \Path_{\mathcal{Y}}$, which is the only path in $\mathcal{Y}$ that has the same ST-trace. There is a sequence of operations
    $
        \beta \unsubsu^2 \beta' \simeq^3 \beta'' \hookleftarrow \delta
    $
    which can be answered only by $\alpha \unsubsu^2 \alpha' \simeq^3 \alpha'' \hookleftarrow \gamma$. But now the extension of $\delta$ to $\delta'$ cannot be matched by $\gamma$.
\end{proof}

\begin{figure}
    \centering
    \begin{tikzpicture}[scale=1.8]
\def\colC{red}
\def\colD{orange}
\def\colA{green!80!black}
\def\colB{blue}
\def\colE{violet}
\def\colF{teal}
\def\colG{olive}
\def\sh{2.3}
\def\sk{4.5}
\def\sl{6.7}
\draw[fill, color=black!15] (0,0) -- (-2,1) -- (0,2) -- (2, 1) -- (0,0);
\node[state] (0) at (0,0) {};
\node[state] (a1) at (-2,1) {};
\node[state] (a2) at (-1,1) {};
\node[state] (a3) at (0,1) {};
\node[state] (a4) at (1,1) {};
\node[state] (a5) at (2,1) {};
\node[state] (1) at (0,2) {};
\node[state] (a0) at (-2,2) {};
\draw[->] (0) edge node[below left] {$b$} (a1);
\draw[->] (0) edge node[above] {$a$} (a2);
\draw[->] (0) edge node[right] {$b$} (a3);
\draw[->] (0) edge node[above] {$a$} (a4);
\draw[->] (0) edge node[below right] {$b$} (a5);
\draw[->] (a1) edge node[left] {$d$} (a0);
\draw[->] (a1) edge node[above left] {$a$} (1);
\draw[->] (a2) edge node[below] {$b$} (1);
\draw[->] (a3) edge node[right] {$a$} (1);
\draw[->] (a4) edge node[below] {$b$} (1);
\draw[->] (a5) edge node[above right] {$a$} (1);
\draw[->, very thick, green!50!black] (-1.92,1)--(-0.04,1.92);
\draw[->, very thick, green!50!black] (-0.04,1)--(-0.04,1.92);
\draw[->, very thick, blue!80!black] (1.92,1)--(0.04,1.92);
\draw[->, very thick, blue!80!black] (0.04,1)--(0.04,1.92);
\draw[->, very thick, red!50!black] (-1.96,1.08)--(0,2.06);
\draw[->, very thick, red!50!black] (1.96,1.08)--(0,2.06);
\node at (-1.2,2.2) {$\mathcal{A}$};
\node[below left, red!80!black] at (a1) {$x$};
\node[below left, red!80!black] at (a2) {$u$};
\node[below right, red!80!black] at (a3) {$y$};
\node[below right, blue!70!gray] at (a4) {$v$};
\node[below right, blue!70!gray] at (a5) {$z$};
\node[above] at (1) {$t$};
\node[red!80!black] at (-1.4,1) {$x_1$};
\node[red!80!black] at (-0.4,1) {$x_2$};
\node[orange!80!black] at (0.4,1) {$x_3$};
\node[blue!70!gray] at (1.4,1) {$x_4$};
\draw[arrow] (0,-0.2) -- (0);
\begin{scope}[shift={(4,0)}] 
\draw[fill, color=black!15] (-1,0) -- (-1,1) -- (0,2) -- (1, 1) -- (1,0) -- (0,1) -- (-1,0);
\node[state] (00) at (-1,0) {};
\node[state] (10) at (0,1) {};
\node[state] (20) at (1,0) {};
\node[state] (01) at (-1,1) {};
\node[state] (11) at (0,2) {};
\node[state] (21) at (1,1) {};
\draw[->,very thick] (00) edge node[below right] {$a$} (10);
\draw[->,very thick] (20) edge node[below left] {$a$} (10);
\draw[->] (01) edge node[above left] {$a$} (11);
\draw[->] (21) edge node[above right] {$a$} (11);
\draw[->] (00) edge node[left] {$c$} (01);
\draw[->] (10) edge node[right] {$c$} (11);
\draw[->] (20) edge node[right] {$c$} (21);
\node at (0,0) {$\mathcal{B}$};
\node[above right] at (10) {$t$};
\end{scope}
\end{tikzpicture}
    \caption{The HDA $\mathcal{X}$ is obtained be gluing $\mathcal{A}$ with two copies of $\mathcal{B}$: $\mathcal{B}_1$ along the edges $x\to t\gets y$ (green) and $\mathcal{B}_2$ along $y\to t\gets z$ (blue). The HDA $\mathcal{Y}$ is obtained by, additionally, gluing an extra copy $\mathcal{B}_3$ along $x\to t\gets z$ (brown).}
    \label{fig:fhp-not-hhp}
\end{figure}
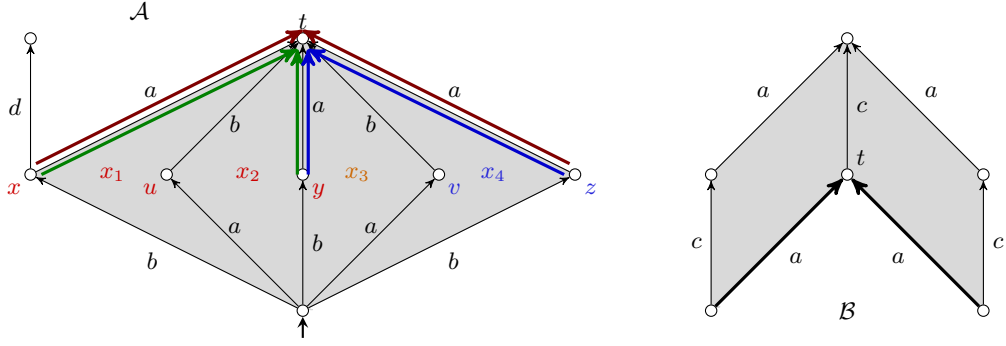
\begin{theorem}    
    ${\bis{hhp}} \subsetneq {\bis{hqhp}}$.
\end{theorem}
\begin{proof}
    Let $\mathcal{X}$, $\mathcal{Y}$ be HDAs illustrated in Figure \ref{fig:fhp-not-hhp}. First we show that they are not hhp-bisimilar. 
    In $\mathcal{Y}$, every path $\alpha$ with ST-trace $b^+b^1$ (there are three such paths, terminating in $x$ or $y$ or $z$) can be extended to a path $\beta$ with ST-trace $b^+b^1c^+c^3a^+a^5$ (there are two possible extension for each $\alpha$, we choose the one crossing $\mathcal{B}_1$ or $\mathcal{B}_3$), which in turn can be transformed using similarities, subsumptions and unsubsumptions into a path $\gamma$ with ST-trace $b^+b^1c^+c^3a^+a^5$ that finally can be restricted to a path $\delta$ with ST-trace $b^+b^1$ and then extended to a (unique) path $\zeta$ with ST-trace $b^+b^1d^+$.

    This cannot be done for the path $\alpha$ leading to $z$ in $\mathcal{X}$: the only possible extension $\beta$ terminates at the top vertex of $\mathcal{B}_2$ and then also $\gamma$ terminates at the same vertex. But the only possible $\delta$ terminates at $y$ or $z$ and thus cannot be extended to $\zeta$.

    Now we prove that $\mathcal{X}$ and $\mathcal{Y}$ are hqhp-bisimilar. 
    Denote $P_0=\Path_{\mathcal{A}}$ and $P_i=\Path_{\mathcal{A}\cup \mathcal{B}_i}\setminus \Path_{\mathcal{A}}$. Thus we have $\Path_{\mathcal{X}}=P_0\sqcup P_1\sqcup P_2$ and $\Path_{\mathcal{Y}}=P_0\sqcup P_1\sqcup P_2\sqcup P_3$.
    Further, we call a path red if it passes through $x$, $x_1$, $u$, $x_2$ or $y$, orange if it passes through $x_3$, and blue if it passes through $v$, $x_4$ or $z$.
    Let $\mathfrak{R}\subseteq \Path_{\mathcal{X}}\times \Path_{\mathcal{Y}}$ be the set of pairs $(\alpha,\beta)$ such that $\ST(\alpha)=\ST(\beta)$ and
    \begin{enumerate}
        \item If $\beta\in \Path(\mathcal{X})$, then $\alpha=\beta$.
        \item If $\beta\in P_3$, then the initial segments of $\alpha$ and $\beta$ that lie in $\mathcal{A}$ are equal.
        \item If $\beta\in P_3$ is red, then $\alpha\in P_1$.
        \item If $\beta\in P_3$ is blue, then $\alpha\in P_2$.
    \end{enumerate}
These conditions imply that a path $\beta$ of $\mathcal{Y}$ is related to exactly two paths of $\mathcal{X}$ if $\beta\in P_3$ and $\beta$ is orange, and to exactly one path of $\mathcal{X}$ otherwise.

     We have to check that the symmetric closure of $\mathfrak{R}$ is an hqhp-bisimulation. Initial paths are trivially related, and related paths have the same ST-traces. The path restriction condition holds by construction.   
    To see that the extension property holds it suffices to note that every red path in $P_0$ can be extended into a path in $P_1$ and every orange and every blue path in $P_0$ can be extended into a path in $P_2$. 

    The subsumption condition holds because
    (a) no further subsumption of an orange path is possible
    (b) applying subsumption to a blue path yields either a blue or an orange path
    (c) applying subsumption to a red path yields either a red or an orange path.
    The similarity condition holds because applying similarity to a path does not change its colour.

Note that the unsubsumption condition would fail, as this can turn an orange path into a red or blue one.
\end{proof}

\section{Generalised bisimulation logic}\label{sec: bis logic}

\begin{definition}\rm
    Let $I$ and $M$ be sets, and let $m_0\in M$ be a distinguished element.
    An \emph{$(I,M)$-system} is a tuple $(A,\alpha_0, \lambda,\{\ltimes_i\}_{i\in I})$, where
    \begin{itemize}
        \item $A$ is a set (of \emph{paths}) with a distinguished element $\alpha_0\in A$,
        \item $\lambda:A\to M$ is a function (a \emph{labelling}) such that $\lambda(\alpha_0)=m_0$,
        \item $\ltimes_i\subseteq A\times A$ is a relation for every $i\in I$.
    \end{itemize}
\end{definition}

\begin{example}
    Let $\mathcal X$ be an HDA\@. Then $(\Path_{\mathcal X}, i_{\mathcal X}, \ST, \{{\ininc},{\subsu}\})$ is a $(\{1,2\},\ST)$-system.
\end{example}

\begin{definition}\rm
\label{d:bisimulation}
    A \emph{bisimulation} between $(I,M)$-systems $A=(A,\alpha_0,\lambda^A, \{\ltimes^A_i\}_{i\in I})$ and $B=(B,\beta_0,\lambda^B, \{\ltimes^B_i\}_{i\in I})$ is a relation ${\mathfrak{R}}\subseteq A\times B$ such that
    \begin{enumerate}
        \item $\alpha_0\mathfrak{R} \beta_0$,
        \item
        \label{i:bisimilation-labels}
        $\alpha\mathfrak{R} \beta \implies \lambda^A(\alpha)=\lambda^B(\beta)$,
        \item
        \label{i:bisimilation-extension}
        $\alpha\mathfrak{R} \beta \;\land\; \alpha\ltimes^A_i \alpha' \implies \exists\; \beta':\; \alpha'\mathfrak{R} \beta' \;\land\; \beta\ltimes^B_i\beta'$,
        \item $\alpha\mathfrak{R} \beta \;\land\; \beta\ltimes^B_i \beta' \implies \exists\; \alpha':\; \alpha'\mathfrak{R} \beta' \;\land\; \alpha\ltimes^A_i\alpha'$,
    \end{enumerate}
    for all $\alpha,\alpha'\in A$, $\beta,\beta'\in B$, $i\in I$.
\end{definition}

\begin{definition}\rm
    The set $\lan{(I,M)}$ of \emph{$(I,M)$-formulas} is generated by
    $$
        F::=
        \top
        \mid
        \lnot F
        \mid
        \bigwedge_{j\in J} F_j
        \mid
        \triangle{m}
        \mid
        \langle i \rangle F,
    $$
    where $m\in M$ and $i\in I$. We write $[i] F := \lnot \langle i \rangle \lnot F$.
\end{definition}

\begin{definition}\rm
    Let $P$ be an $(I,M)$-system.
    The \emph{satisfaction relation} ${\models}\subseteq P\times \lan{(I,M)}$ is defined by structural induction:
    \begin{itemize}
        \item $\alpha\models\top$;
        \item $\alpha\models \lnot F \iff \lnot(\alpha\models F)$;
        \item $\alpha\models \bigwedge_{j\in J}F_j \iff \forall\;j\in J:\; \alpha\models F_j$;
        \item $\alpha\models \triangle m \iff \lambda(\alpha)=m$;
        \item $\alpha\models \langle i \rangle F \iff \exists \beta\in P:\; \alpha \ltimes_i \beta\;\land\; \beta\models F$.
    \end{itemize}
    A formula holds for an $(I,M)$-system $A=(A,\alpha_0,\lambda^A, \{\ltimes^A_i\}_{i\in I})$ iff it holds for $\alpha_0$.
\end{definition}

\begin{definition}\rm
    $(I,M)$-systems $P$ and $Q$ are \emph{logically equivalent} if $\alpha^P_0\models F \iff \alpha^Q_0\models F$ for every $F\in \lan{(I,M)}$.
\end{definition}

\begin{theorem}\rm\label{modal char}
    $(I,M)$-systems $P$ and $Q$ are bisimilar iff they are logically equivalent.
\end{theorem}
\begin{proof}
  ($\Rightarrow$) Assume that $\mathfrak{R}$ is a bisimulation between $P$ and $Q$.
  By induction on $F$, we show that for all $F\in\lan{(I,M)}$ the formula
    $$
        E(F)::=
        \forall\alpha\in P,\;
        \forall\beta\in Q:\;
        \alpha\mathfrak{R} \beta 
        \implies 
        (\alpha\models F \iff \beta\models F).
    $$
    is satisfied. It is clear that $E(\top)$, $E(F)\implies E(\lnot F)$ and
    $
        (\forall_{j\in J}:\; E(F_j))
        \implies 
        E(\bigwedge_{j\in J}F_j)
    $.
    The formula $E(\triangle P)$ follows from condition \ref{d:bisimulation}.\ref{i:bisimilation-labels}.
    Assume that $\alpha\in P$, $\beta\in Q$ and $\alpha\mathfrak{R} \beta$. Then 
    \begin{align*}
        \alpha\models \langle i \rangle F
        \iff
        &\exists_{\alpha'\in A}\;\alpha\ltimes^A_i\alpha'\;\land\; \alpha'\models F
        \\
        \implies
        &
        \exists_{\alpha'\in A}\;
        \exists_{\beta'\in B}\;
        \alpha'\mathfrak{R}\beta'\;\land\;
        \beta\ltimes^B_i\beta'\;\land\;
        \alpha'\models F
        \tag{by \ref{d:bisimulation}.\ref{i:bisimilation-extension}}
        \\
        \implies
        &
        \exists_{\beta'\in B}\;
        \beta\ltimes_i\beta'\;\land\;
        \beta'\models F\tag{by inductive assumption}
        \\
        \iff
        &
        \beta\models \langle i \rangle F.
    \end{align*}
    Similarly we show that $\beta\models \langle i \rangle F$ implies $\alpha\models \langle i \rangle F$.
    Consequently, $E(F)\implies E(\langle i \rangle F)$.

    ($\Leftarrow$)
    Assume that $A$ and $B$ and logically equivalent. Define a relation ${\mathfrak{R}}\subseteq
    A\times B$ by
    $$
        \alpha\mathfrak{R} \beta
        \iff 
        \forall_{F\in\lan{(I,M)}}
        (\alpha\models F \iff \beta\models F).
    $$
    By assumption, $\alpha_0\mathfrak{R}\beta_0$. Moreover, if  $\alpha\mathfrak{R}\beta$ then for all $m\in M$ one has $\alpha\models \triangle m$ iff $\beta\models \triangle m$, and thus $\lambda(\alpha)=m$ iff $\lambda(\beta)=m$. Hence also the second clause of Def.~\ref{d:bisimulation} is satisfied.
    Assume that $\alpha\ltimes^A_i\alpha'$ and $\alpha\mathfrak{R}\beta$ for $\alpha,\alpha'\in A$ and $\beta\in B$. We need to show that there exists $\beta'\in B$ such that $\beta\ltimes_i\beta'$ and $\alpha'\mathfrak{R}\beta'$. Denote
    $$
        C=\{\beta'_j\}_{j\in J} = 
        \{\beta'\in B\mid \beta\ltimes_i \beta'\}.
    $$
    If there exists $\beta'_j\in B$ such that $\alpha'\models F \iff \beta'_j\models F$ for all $F$, then we have $\alpha'\mathfrak{R} \beta'_j$. Otherwise, for every $j\in J$ there exists $F_j\in\lan{(I,M)}$ such that $\alpha'\models F_j$ and $\beta'_j\models \lnot F_j$. Finally,
    $$
        \alpha \models \langle i \rangle\bigwedge_{j\in J}F_j
        \qquad\text{and} \qquad
        \beta \models \lnot\; \langle i \rangle\bigwedge_{j\in J}F_j,
    $$
    a contradiction. 
\end{proof}

\section{Modal characterisations}\label{sec: Modal char}

Depending on the semantic equivalence we want to focus on,
every HDA $\mathcal X$ defines an $(I,M)$-system $(A,\alpha_0, \lambda,\{\ltimes_i\}_{i\in I})$
  with $M :=$ the set of ST traces, $A:=\Path_{\mathcal{X}}$, $\alpha_0 := i_{\mathcal X}$,
  $\lambda := \ST:\Path_{\mathcal{X}}\to M$, and
    \begin{itemize}
    \item (ST) $I:=\{{\ininc}\}$,
    \item (hp) $I:=\{{\ininc}, {\subsu^{\ell}}, {\sqsupseteq^\ell}, {\simeq^\ell} \mid \ell\in\IN\}$,
    \item (hhp) $I:=\{{\ininc}, {\subsu^{\ell}}, {\sqsupseteq^\ell}, {\simeq^\ell}, {\hookleftarrow} \mid \ell\in\IN\}$,
    \item (qhp) $I:=\{{\ininc}, {\subsu^{\ell}}, {\simeq^\ell} \mid \ell\in\IN\}$,
    \item (shp) $I:=\{{\ininc}, {\subsu^{\ell}} \mid \ell\in\IN\}$,
    \end{itemize}
    etc., with $\ltimes_{\ininc} := {\ininc} \subseteq \Path_{\mathcal{X}}\times\Path_{\mathcal{X}}$,
    and likewise for the other elements of $I$.

    By Thm.~\ref{modal char}, this immediately yields a modal characterisation for each of the eight semantic equivalences above. In the case of hhp-bisimilarity, for instance, the set of formulae is generated by
   $$
        F::=
        \top
        \mid
        \lnot F
        \mid
        \bigwedge_{j\in J} F_j
        \mid
        \triangle{m}
        \mid
        \langle \ininc \rangle F
        \mid
        \langle {\subsu^{\ell}} \rangle F
        \mid
        \langle {\sqsupseteq^\ell} \rangle F
        \mid
        \langle {\simeq^\ell} \rangle F
        \mid
        \langle {\hookleftarrow} \rangle F,
    $$
        with $m$ an ST trace .
We address this logic as $\lan{hhp}$, and similarly for its fragments.

\begin{example}
  We can show that the HDAs of Fig \ref{fi: ST bisimilar HDA} are not shp-bisimilar by the distinguishing $\lan{shp}$ formula $F:=  \langle \ininc \rangle \textcolor{red}{\big(}\Lambda(a^+ a^1 b^+ b^3) \wedge \neg \langle \subsu^2 \rangle \top\textcolor{red}{\big)}$. Here $\langle \ininc \rangle\varphi$ is a formula that holds for an HDA iff it has a path that satisfies $\varphi$, namely an extension of the initial path.  The formula $F$ says that there exists a path with ST-trace $a^+ a^1 b^+ b^3$ (thus executing $a$ and $b$ in succession) that has no subsumption at position 2.
It holds for the right-hand HDA, but not for the left one.
\end{example}

\begin{example}
  The HDAs of Fig.~\ref{fig:projected HDA} are not qhp-bisimilar, as witnessed by the $\lan{qhp}$ formula $F:=  \langle \ininc \rangle \textcolor{red}{\big(}\Lambda(a^+c^+a^1c^2b^+d^+b^5d^6)\wedge \langle \subsu^4 \rangle \top \wedge \langle \simeq^3\rangle  \langle \simeq^5\rangle \langle \subsu^4 \rangle \top \textcolor{red}{\big)}$.
   This formula says that there exists a path with the indicated ST-trace,
   that allows a subsumption at position 4, and also allows similarities at positions 3 and 5, following by a subsumption at position 4.
It holds for the right-hand HDA of Fig.~\ref{fig:projected HDA}, but not for the left one.
\end{example}

\begin{example}\label{ex:absorption-law}
The two HDAs of Fig.~\ref{fig:absorption law}
are distinguished by the following $\lan{hhp}$-formula:
$$F:= \langle \ininc \rangle \textcolor{red}{\big(}\Lambda(a^+b^+) \wedge \langle \hookleftarrow \rangle (\Lambda(a^+) \wedge \neg \langle \hookrightarrow \rangle \Lambda(a^+ c^+)) \wedge \langle \simeq^1 \rangle \langle \hookleftarrow \rangle  (\Lambda(b^+) \wedge \neg \langle \hookrightarrow \rangle \Lambda(b^+ c^+))\textcolor{red}{\big)}\;.$$
It says that there exists a path in which first $a$ and then $b$ start, such that
(1) the restriction of this path to the start of $a$ cannot be extended with the start of $c$, and (2) after swapping the starts of $a$ and $b$, the restriction of this path to the start of $b$ cannot be extended with the start of $c$. This formula holds only for the right-hand side of Fig.~\ref{fig:absorption law}. The formulas sits even in $\lan{reST}$.
\end{example}

\begin{example}
  The HDAs $\mathcal X_1$ and $\mathcal X_2$ of Fig.~\ref{fig:hp-not-qhp2} can be distinguished with the formula $[\hookrightarrow]\textcolor{red}{\big(} \Lambda(a^+) \rightarrow  \langle \hookrightarrow \rangle \textcolor{blue}{\big(}\Lambda(a^+ a+) \wedge  \langle \simeq \rangle \langle \hookleftarrow \rangle \textcolor{magenta}(\Lambda(a^+) \wedge  \langle \hookrightarrow \rangle \Lambda(a^+ b+) \textcolor{magenta})\textcolor{blue}{\big)}\textcolor{red}{\big)}$, which holds for $\mathcal X_1$ but not for $\mathcal X_2$. It says that any path with label $a^+$ can be extended to have label $a^+ a^+$, then undergo a similarity swap, following by a restriction to the label $a^+$, so that the result of that can be extended to have label $a^+ b^+$.
  This is an $\lan{reST}$-formula, showing that $\mathcal X_1 \not\bis{reST} \mathcal X_2$.

  They can also be distinguished by
  $$F:=  [\hookrightarrow] \textcolor{red}{\big(} \Lambda(a^+a^1 a^+a^3 b^+) \rightarrow  \langle \subsu^2 \rangle  \langle \simeq^1 \rangle\langle \simeq^3\rangle \langle \unsubsu^2 \rangle\langle \subsu^4 \rangle \Lambda(a^+a^1 a^+ b^+ a^3)\textcolor{red}{\big)}\;.$$
It says that on any path with the indicated ST-trace one can apply the operators $\subsu^2$, $\simeq^1$, $\simeq^3$, $\unsubsu^2$ and $\subsu^4$, in that order, to obtain a path with label $a^+a^1 a^+ b^+ a^3$. Again it holds for $\mathcal X_1$ but not for $\mathcal X_2$.
This is an $\lan{hp}$-formula, showing that $\mathcal X_1 \not\bis{hp} \mathcal X_2$.
\end{example}

\bibliography{hda}
\end{document}